\UseRawInputEncoding
\documentclass[aps,pra,twocolumn,longbibliography]{revtex4-1}
\usepackage[colorlinks=true, citecolor=red, urlcolor=blue ]{hyperref}
 \usepackage[utf8]{inputenc}
\usepackage{graphicx}   
\usepackage{dcolumn}    
\usepackage{bm}         
\usepackage{hyperref}   
\usepackage{xcolor}     
\usepackage{amsmath}
\usepackage{relsize}
\usepackage{comment}
\usepackage[caption=false]{subfig} 
\usepackage{ragged2e}
\usepackage{booktabs}
\usepackage{indentfirst}
\usepackage{epsfig}
\usepackage{amsthm} 
 \usepackage{amssymb}
 \usepackage{subfig}
 \usepackage{mathtools}
\usepackage{mathrsfs} 
\usepackage[normalem]{ulem}
\usepackage{amsthm}

\newcommand{\SuM}[1]{{\color{black}{#1}}}
\newcommand{\SaM}[1]{{\color{black}{#1}}}
\newcommand{\residual}[1]{{\color{red}{#1}}}

\newcommand{\sys}{\mathbb{S}}

\newcommand{\rhi}{\rho_i}
\newcommand{\rhAi}{\rho^A_i}
\newcommand{\rhBi}{\rho^B_i}
\newcommand{\rhti}{\tilde{\rho}_i}

\newcommand{\rhf}{\rho_f}

\newcommand{\gbi}{\gamma_{\beta,i}}
\newcommand{\gbAi}{\gamma^A_{\beta,i}}
\newcommand{\gbBi}{\gamma^B_{\beta,i}}
\newcommand{\gbf}{\gamma_{\beta,f}}
\newcommand{\gbAf}{\gamma^A_{\beta,f}}
\newcommand{\gbBf}{\gamma^B_{\beta,f}}
\newcommand{\chit}{\tilde{\chi}}

\newcommand{\taud}{\tau_d}
\newcommand{\tauc}{\tau_c}
\newcommand{\tautd}{\tilde{\tau}_d}
\newcommand{\tautc}{\tilde{\tau}_c}
\newcommand{\EnAB}{\mathfrak{E}_{AB}}
\newcommand{\EntAB}{\tilde{\mathfrak{E}}_{AB}}
\newcommand{\rhE}{\rho_{\mathcal{E}}}
\newcommand{\rhtE}{\tilde{\rho}_{\mathcal{E}}}
\newcommand{\rhS}{\rho_{\mathcal{S}}}
\newcommand{\rhtS}{\tilde{\rho}_{\mathcal{S}}}
\newcommand{\rhAj}{\rho^A_j}
\newcommand{\rhBj}{\rho^B_j}
\newcommand{\jt}{\tilde{j}}

\newcommand{\rj}{r_{j}}
\newcommand{\rtj}{\tilde{r}_{\tilde{j}}}
\newcommand{\tauAdj}{\tau^A_{d_j}}
\newcommand{\tauBdj}{\tau^B_{d_j}}
\newcommand{\tauAcj}{\tau^A_{c_j}}
\newcommand{\tauBcj}{\tau^B_{c_j}}

\newcommand{\rhdj}{\rho_{d_j}}
\newcommand{\rhcj}{\rho_{c_j}}
\newcommand{\rhtdj}{\tilde{\rho}_{d_{\tilde{j}}}}
\newcommand{\rhtcj}{\tilde{\rho}_{c_{\tilde{j}}}}
\newcommand{\rhstar}{\rho_{*}}

\newcommand{\Hi}{H(t_i)}
\newcommand{\HAi}{H_A(t_i)}
\newcommand{\HBi}{H_B(t_i)}
\newcommand{\Hf}{H(t_f)}
\newcommand{\HAf}{H_A(t_f)}
\newcommand{\HBf}{H_B(t_f)}
\newcommand{\Hint}{H_{ {\rm int}}}

\newcommand{\Pht}{\Phi_{t}}
\newcommand{\Phtf}{\Phi_{t_f}}
\newcommand{\Phtil}{\tilde{\Phi}}

\newcommand{\Pii}{\Pi^i}
\newcommand{\Piil}{\Pi^i_l}
\newcommand{\PiilA}{\Pi^i_{l_A}}
\newcommand{\PiilB}{\Pi^i_{l_B}}

\newcommand{\Pifk}{\Pi^f_k}
\newcommand{\PifkA}{\Pi^f_{k_A}}
\newcommand{\PifkB}{\Pi^f_{k_B}}

\newcommand{\pin}{p^{i}}
\newcommand{\pinl}{p^{i}_l}
\newcommand{\pinlb}{p^{i}_\textbf{l}}

\newcommand{\pfink}{p^{f}_k}
\newcommand{\pfinkb}{p^{f}_\textbf{k}}

\newcommand{\ptink}{\tilde{p}^{i}_k}
\newcommand{\ptinkb}{\tilde{p}^{i}_\textbf{k}}

\newcommand{\ptfinl}{\tilde{p}^{f}_l}
\newcommand{\ptfinlb}{\tilde{p}^{f}_\textbf{l}}
\newcommand{\PG}{P_{\G}}
\newcommand{\PGt}{P_{\Gt}}
\newcommand{\Pc}{P_{coh}}
\newcommand{\pcoh}{p_{\text{coh}}}
\newcommand{\pent}{p_{\text{ent}}}
\newcommand{\Pe}{P_{ent}}

\newcommand{\s}{\sigma}

\newcommand{\sfk}{\sigma^{(f)}_k}
\newcommand{\slk}{\sigma_{l,k}}
\newcommand{\st}{\tilde{\sigma}}

\newcommand{\stfl}{\tilde{\sigma}^{(f)}_l}
\newcommand{\slbkb}{\sigma_{\textbf{l},\textbf{k}}}
\newcommand{\sfkb}{\sigma^{(f)}_\textbf{k}}
\newcommand{\stflb}{\tilde{\sigma}^{(f)}_\textbf{l}}
\newcommand{\SB}{\Sigma}

\newcommand{\SBil}{\Sigma^{(i)}_l}
\newcommand{\SBfk}{\Sigma^{(f)}_k}

\newcommand{\SBlk}{\Sigma_{l,k}}

\newcommand{\SBtfl}{\tilde{\Sigma}^{(f)}_l}
\newcommand{\SBtik}{\tilde{\Sigma}^{(i)}_k}
\newcommand{\Th}{\Theta}

\newcommand{\Thil}{\Theta^{(i)}_l}

\newcommand{\Thtik}{\tilde{\Theta}^{(i)}_k}

\newcommand{\Thfk}{\Theta^{(f)}_k}

\newcommand{\Thtfl}{\tilde{\Theta}^{(f)}_l}
\newcommand{\Thlk}{\Theta_{l,k}}
\newcommand{\stot}{s_{\rm tot}}
\newcommand{\stotlk}{s^{l,k}_{\rm tot}}
\newcommand{\stotlbkb}{s^{\textbf{l},\textbf{k}}_{\rm tot}}
\newcommand{\scorlk}{s^{l,k}_{\rm corr}}
\newcommand{\scorlbkb}{s^{\textbf{l},\textbf{k}}_{\rm corr}}

\newcommand{\Xiilb}{\Xi^{(i)}_\textbf{l}}

\newcommand{\Xitikb}{\tilde{\Xi}^{(i)}_\textbf{k}}

\newcommand{\Xifkb}{\Xi^{(f)}_\textbf{k}}

\newcommand{\Xitflb}{\tilde{\Xi}^{(f)}_\textbf{l}}
\newcommand{\Xilbkb}{\Xi_{\textbf{l},\textbf{k}}}
\newcommand{\Lamlbkb}{\Lambda_{\textbf{l},\textbf{k}}}

\newcommand{\Lamfkb}{\Lambda^{(f)}_{\textbf{k}}}
\newcommand{\Lamilb}{\Lambda^{(i)}_{\textbf{l}}}
\newcommand{\Lamtikb}{\tilde{\Lambda}^{(i)}_{\textbf{k}}}
\newcommand{\Lamtflb}{\tilde{\Lambda}^{(f)}_{\textbf{l}}}
\newcommand{\SSfkbj}{\mathcal{S}^{(f)}_{\textbf{k},j}}
\newcommand{\SSilbj}{\mathcal{S}^{(i)}_{\textbf{l},j}}
\newcommand{\SStflbj}{\tilde{\mathcal{S}}^{(f)}_{\textbf{l},\tilde{j}}}
\newcommand{\SStikbj}{\tilde{\mathcal{S}}^{(i)}_{\textbf{k},\tilde{j}}}
\newcommand{\Psiilb}{\Psi^{(i)}_\textbf{l}}
\newcommand{\Psifkb}{\Psi^{(f)}_\textbf{k}}
\newcommand{\Psitikb}{\tilde{\Psi}^{(i)}_\textbf{k}}
\newcommand{\Psitflb}{\tilde{\Psi}^{(f)}_\textbf{l}}

\newcommand{\B}{\beta}
\newcommand{\ti}{t_i}
\newcommand{\tf}{t_f}
\newcommand{\Ei}{E^i}
\newcommand{\Eil}{E^i_l}
\newcommand{\EilA}{E^i_{l_A}}
\newcommand{\EilB}{E^i_{l_B}}

\newcommand{\Ef}{E^f}
\newcommand{\Efk}{E^f_k}
\newcommand{\EfkA}{E^f_{k_A}}
\newcommand{\EfkB}{E^f_{k_B}}

\newcommand{\zbi}{\mathcal{Z}_{\beta,i}}
\newcommand{\zbAi}{\mathcal{Z}^A_{\beta,i}}
\newcommand{\zbBi}{\mathcal{Z}^B_{\beta,i}}
\newcommand{\zbf}{\mathcal{Z}_{\beta,f}}

\newcommand{\Elk}{E_{l,k}}
\newcommand{\Elbkb}{E_{\textbf{l},\textbf{k}}}

\newcommand{\aAj}{a^A_j}
\newcommand{\aBj}{a^B_j}
\newcommand{\cAj}{c^A_j}
\newcommand{\cBj}{c^B_j}
\newcommand{\atAj}{\tilde{a}^A_{\tilde{j}}}
\newcommand{\atBj}{\tilde{a}^B_{\tilde{j}}}

\newcommand{\D}{\Delta}
\newcommand{\tr}{\text{Tr}}
\newcommand{\G}{\Gamma}
\newcommand{\Gt}{\tilde{\Gamma}}

\newcommand{\lA}{l_A}
\newcommand{\lB}{l_B}
\newcommand{\kA}{k_A}
\newcommand{\kB}{k_B}

\newcommand{\ctl}{\tilde{c}}
\newcommand{\atl}{\tilde{a}}

\newcommand{\lamt}{\tilde{\lambda}}

\newcommand{\lb}{\textbf{l}}
\newcommand{\kb}{\textbf{k}}
\newcommand{\corr}{\Psi}

\newtheorem{theorem}{Theorem}

\theoremstyle{definition}
\newtheorem{definition}[theorem]{Definition}
\AtBeginEnvironment{definition}{\vspace{2pt}\noindent\ignorespaces}

\newtheorem{lemma}[theorem]{Lemma}

\begin{document}

\title{
Resource-resolved quantum fluctuation theorems in end-point measurement scheme
}

\author{Sukrut Mondkar}
\email{sukrutmondkar@gmail.com}
\affiliation{Harish-Chandra Research Institute, A CI of Homi Bhabha National Institute, Chhatnag Road, Jhusi,
Prayagraj (Allahabad) 211019, India}

\author{Sayan Mondal}
\email{sayanmondal96sbs@gmail.com}
\affiliation{Harish-Chandra Research Institute, A CI of Homi Bhabha National Institute, Chhatnag Road, Jhusi,
Prayagraj (Allahabad) 211019, India}

\author{Ujjwal Sen}
\email{ujjwalsen0601@gmail.com}
\affiliation{Harish-Chandra Research Institute, A CI of Homi Bhabha National Institute, Chhatnag Road, Jhusi,
Prayagraj (Allahabad) 211019, India}


\begin{abstract}


Fluctuation theorems provide universal constraints on nonequilibrium energy and entropy fluctuations, making them a natural framework to assess how and to what extent quantum resources become thermodynamically relevant. We develop a unified framework for incorporating a generic quantum resource, including athermality, quantum coherence, and entanglement, into fluctuation theorems. We work within the end-point measurement  scheme, which avoids an initial  energy measurement and allows quantum resources in the initial state to affect nonequilibrium energy statistics. We derive a family of  
quantum fluctuation theorems, including generalized Jarzynski equalities and Crooks-type fluctuation relations, in which corrections decompose into 
resource-resolved contributions. For single systems, we introduce the concept of weight of athermality, and combine it with the weight of coherence to isolate distinct thermodynamic effects of these quantum resources.
For bipartite systems, we furthermore obtain two families of entanglement-resolved fluctuation theorems using an appended correlation operator 
and the best separable approximation, respectively.
Finally, we introduce the concepts of  
 coherence- and entanglement-fluctuation distances, as Kullback–Leibler divergences,  
 which
 quantify the thermodynamic relevance of quantum resources in 
a process-dependent and operational manner.

\end{abstract}

\maketitle


\section{Introduction}

The laws of thermodynamics impose fundamental constraints on the behavior of physical systems, governing energy exchange, irreversibility, and entropy production in macroscopic systems close to equilibrium. Traditionally, these laws are formulated at the ensemble-averaged level, where fluctuations around mean values are negligible~\cite{Reif2016, Landau:1980mil, Callen1991, pathria2016}. 
However, with the advent of experimental platforms capable of manipulating and probing systems at mesoscopic and microscopic scales~\cite{Liphardt2002, PhysRevLett.89.050601, Collin2005, 
PhysRevLett.113.140601, PhysRevX.7.021051, 5lp2-9sps, PhysRevLett.127.180603, Zhang2018, PhysRevResearch.2.023327, HernndezGmez2021, PhysRevA.101.052113}, the role of fluctuations has become unavoidable.
Stochastic thermodynamics provides a unified framework for describing small systems driven far from equilibrium, in which thermodynamic quantities such as work, heat, and entropy production become intrinsically stochastic variables~\cite{Seifert2008, Seifert_2012, Sekimoto2010, PhysRevE.82.011143, RevModPhys.81.1665}. A remarkable outcome of this framework is the discovery of fluctuation theorems (FTs)~\cite{PhysRevLett.74.2694, PhysRevLett.71.2401, JorgeKurchan_1998, Lebowitz1999, Harris_2007, Peliti2021,Seifert_2025}, which place universal constraints on the full probability distributions of these fluctuating quantities, valid arbitrarily far from equilibrium. Such FTs include the celebrated Jarzynski equality~\cite{PhysRevLett.78.2690} and the Crooks fluctuation theorem~\cite{PhysRevE.60.2721}.

Extending stochastic thermodynamics to the quantum regime introduces additional conceptual and operational challenges.
The most widely used framework for defining work and entropy production in quantum systems is based on the two-point measurement (TPM) protocol~\cite{PhysRevE.75.050102, talkner2007, RevModPhys.81.1665,RevModPhys.83.771, 10.1116/5.0079886, Strasberg2022}. In this approach, projective measurements of the system energy are performed at the beginning and at the end of the process, and work is defined as the difference between the measured outcomes. The TPM scheme has enabled the formulation of quantum versions of fluctuation theorems and has provided a consistent operational definition of work compatible with classical limits~\cite{tasaki2000jarzynskirelationsquantumsystems, PhysRevLett.90.170604, Kurchan:2000rzb, RevModPhys.83.771, Baumer2018}.

Despite its success, the TPM protocol has recently come under scrutiny. This protocol fundamentally precludes an assessment of how genuinely quantum resources, such as coherence and entanglement, affect nonequilibrium thermodynamics. 
The initial projective measurement destroys quantum coherence and entanglement present in the system’s initial state, thereby eliminating genuinely quantum features before the dynamics even begins~\cite{RevModPhys.83.771, PhysRevLett.118.070601, Micadei:2019wzk}. As a consequence, the TPM framework is blind to the thermodynamic role of quantum resources.
Deriving fluctuation theorems that remain sensitive to these resources therefore requires frameworks that go beyond TPM while retaining operational meaning. Several frameworks have been proposed in this direction, including formulations based on quasi-probability distributions~\cite{PhysRevLett.120.040602, PRXQuantum.1.010309, PRXQuantum.5.030201, Jae:2025lff, Li:2025kfx}, Bayesian networks~\cite{PhysRevLett.124.090602, PhysRevLett.127.180603} and the end-point measurement (EPM) scheme~\cite{PhysRevA.104.L050203, Hernandez-Gomez:2022xor, Gianani:2022fyp, Artini:2025lqf, 
Jae:2025lff}.

The EPM scheme provides a natural operational setting for deriving quantum fluctuation theorems. By avoiding an initial projective energy measurement and relying instead on virtual initial energies and a single final projective measurement, the EPM protocol allows initial quantum resources to influence the statistics of nonequilibrium energy changes. As a result, EPM offers a flexible framework in which distinct quantum resources can give rise to distinct, physically interpretable corrections to fluctuation relations.

In this work, we develop a unified framework for deriving a zoo of quantum fluctuation theorems that systematically resolve the thermodynamic roles of different quantum resources within the EPM scheme. We obtain four distinct categories of fluctuation theorems. First, for single quantum systems, we revisit EPM-based Jarzynski equalities and Crooks-type detailed fluctuation theorems using the decomposition of an initial coherent state into a thermal part and a {coherence operator}~\cite{PhysRevA.104.L050203,Hernandez-Gomez:2022xor}. Second, we introduce a fully resource-theoretic formulation for single systems based on convex decompositions using the weight of athermality and the weight of coherence, leading to refined fluctuation theorems in which classical uncertainty, athermality, and coherence contributions are cleanly separated. Here, we introduce the weight of athermality as a new operational measure of athermality.
Third, for bipartite quantum systems, we derive fluctuation theorems using a {correlation-operator} decomposition that captures total correlations, including entanglement, beyond product structure. Finally, we present a fully operational and entanglement-resolved family of quantum fluctuation theorems by combining the best separable approximation with local decompositions based on the weight of athermality and the weight of coherence. Together, these four types of FTs, provide a systematic and operational classification of quantum fluctuation theorems, clarifying when and how coherence, athermality, and entanglement modify nonequilibrium energy statistics.  {In Fig.~\ref{fig:1}, we summarize the four families of FTs in a pictorial form.}

Beyond modified fluctuation relations, 
we introduce new measures
that quantify the thermodynamic relevance of coherence and entanglement. These measures capture how strongly a given quantum resource affects 
{non-equilibrium energy fluctuations}, providing a direct operational link between quantum resource theories and non-equilibrium thermodynamics. 

The structure of the paper is as follows. In Sec.~\ref{sec:EPM-review}, we briefly review the {EPM} scheme and the resulting corrections to the Jarzynski equality and entropy fluctuation theorems. In Sec.~\ref{sec:EPM-FT}, we apply our decomposition framework to single-system scenarios and analyze the effects of athermality and coherence. In Sec.~\ref{sec:EPM-FT-biparty}, we extend the analysis to bipartite systems, considering both correlation-based {correlation operator} and best separable decompositions. In Sec.~\ref{sec:new-measures}, we introduce new measures that quantify the thermodynamic impact of coherence and entanglement. Finally, Sec.~\ref{sec:conclusion} summarizes our findings and outlines possible future directions.


\begin{figure*}
    \centering
    \includegraphics[width=0.7\linewidth]{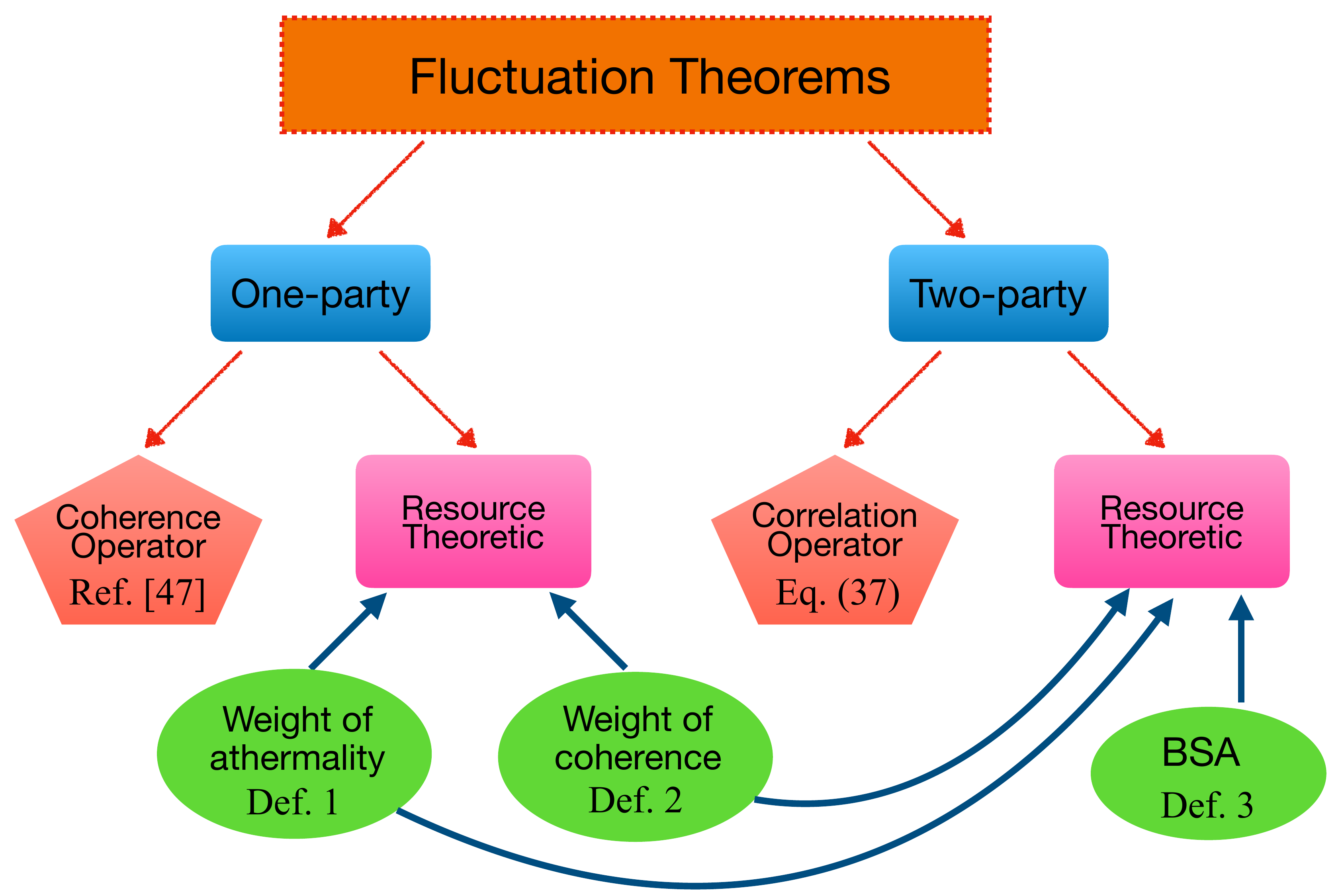}
    \caption{\emph{\textbf{Zoo of fluctuation theorems.}} 
The figure provides a schematic classification of all fluctuation theorems (FTs) derived in this work. We organize the results according to the system structure: one-party (single) versus two-party (bipartite) quantum systems and the manner in which the initial resourceful quantum state is decomposed.
For one-party scenarios, we distinguish between FTs obtained using the previously employed coherence-operator decomposition (orange pentagon), which leads to EPM Jarzynski equality and entropy-production relations derived in Refs.~\cite{PhysRevA.104.L050203, Hernandez-Gomez:2022xor}, and a resource-theoretic decomposition (pink boxes) based on the weight of athermality (Def.~\ref{def:athermality}) and weight of coherence (Def.~\ref{def:coherence}). The latter yields refined FTs in which classical uncertainty, athermality, and coherence contributions are cleanly separated (Eq.~\eqref{eq:Jarzynski-coh}, Eq.~\eqref{eq:DFT-entropy-single-party} and Eq.~\eqref{eq:Theta-Sigma-Final}). 
For two-party case, we first consider a correlation-operator decomposition (orange pentagon), which captures total correlations  
leading to correlation-corrected FTs (Eq.~\eqref{eq:Jarzynski-EnAB-biparty} and Eq.~\eqref{eq:DFT-entropy-bi-party} with Eq.~\eqref{eq:DeltaPsi-residual}). We then introduce a fully operational and resource-theoretic formulation based on the Best Separable Approximation (Def.~\ref{def:BSA}), which yields entanglement-resolved FTs (Eq.~\eqref{eq:Jarzynski-ent} and Eq.~\eqref{eq:DFT-entropy-bi-party} with Eq.~\eqref{eq:Corre-corrections}).
Together, the four branches shown in the diagram constitute a unified framework for systematically isolating and quantifying the thermodynamic roles of coherence, athermality, and entanglement in non-equilibrium quantum fluctuation relations within the EPM protocol.
    }
    \label{fig:1}
\end{figure*}

\section{Review of the End-Point Measurement scheme}\label{sec:EPM-review}

The EPM scheme~\cite{PhysRevA.104.L050203, Hernandez-Gomez:2022xor} provides a framework for analyzing the statistics of non-equilibrium energy fluctuations without requiring a projective energy measurement on the initial quantum state, unlike in the TPM protocol. By avoiding the collapse of the initial state, EPM allows one to analyze how genuine quantum features such as coherence and entanglement influence the statistics of non-equilibrium energy fluctuations.

%

\subsection{The protocol}\label{sec:EPM-rev:subsec:proto}

Let $\sys$ be a $d$-dimensional quantum system whose dynamics over the time interval $[\ti, \tf]$ are described by a one-parameter family of completely positive, trace-preserving (CPTP) maps 
$\Pht$, which evolve an initial state $\rhi$ to a final state $\rhf = \Phtf[\rhi]$. At the final time $t_f$, the system is subjected to a projective measurement in the energy eigenbasis. This maps the initial state $\rhi$ to an outcome $\Pifk$, where $\Pifk := | \Efk \rangle \langle \Efk  |$ is the projection operator corresponding to the $k$-th eigenstate $| \Efk \rangle $ of the final Hamiltonian $\Hf = \sum_k \Efk \Pifk$. In other words, the EPM protocol generates trajectories $\mathcal{T}_k (\rhi):\rhi \rightarrow \Pifk$. We will suppress the ${t}$ subscript in $\Pht$ from here on and represent the dynamical map as $\Phi$. 
Given the initial Hamiltonian $\Hi  = \sum_l \Eil \Piil$, one can identify the spectral probabilities $\tr \left(\rhi \Piil \right)$ 
associated with the state $\rhi$ by preparing an ensemble of identical systems in this state and performing projective energy measurements in the eigenbasis of $\Hi$ on only a subset of them, while leaving the remaining systems in the ensemble undisturbed~\cite{PhysRevA.104.L050203}.
Since no projective energy measurement is carried out at the initial time, the outcomes of the EPM scheme are inherently stochastic with respect to the {virtual} initial energy $\Ei$, that the system would have possessed had such a measurement been performed. Consequently, the energy difference $\D E = \Ef - \Ei$ is treated as a random variable. The uncertainty in $\Ei$ stems from the fact that these values are only {assigned} as outcomes of hypothetical projective measurements, which do not collapse the quantum state. Thus, the actual projective measurement at the final time $\tf$ remains uncorrelated with these virtual initial outcomes, leading to statistical independence of the final projective measurement.

The \SuM{EPM} probability distribution, $\Pc(\D E)$, associated with the random variable $\D E$,  is defined as
\begin{equation}
    \Pc(\D E) := \sum_{l,k} p(\Eil, \Efk ) \delta(\D E - \D \Elk) 
\end{equation}
where $\D \Elk  =\Efk - \Eil$ and $p(\Eil , \Efk )$ is the joint probability of virtually measuring intial state $| \Eil  \rangle$ and measuring final state $| \Efk  \rangle$. The statistical independence of the final projective energy measurements from the initial virtual ones implies that 
\begin{align}\label{eq:pcoh}
    p(\Eil , \Efk ) = \pinl  \pfink  &= \tr(\rhi \Piil) \tr(\Phi[\rhi] \Pifk) 
    =: \pcoh^{l,k}
\end{align}
where $\pinl  := \tr(\rhi \Piil)$ and $\pfink  := \tr(\Phi[\rhi] \Pifk)$ are the probabilities associated with the initial virtual and final energy measurements, and $\pcoh^{l,k}$ is the joint probability distribution satisfying $\sum_{l,k} \pcoh^{l,k} = 1$. 

Crucially, the average energy change obtained from the EPM distribution,
\begin{equation}
    \langle \D E \rangle_{\Pc} := \int d \D E \Pc (\D E) \D E
\end{equation}
exactly reproduces the average energy variation generated by the CPTP map $\Phi$~\cite{PhysRevA.104.L050203},
\begin{equation}
    \langle \D E \rangle = \tr\left(\Hf \rhf \right) - \tr \left( \Hi \rhi \right)
\end{equation}
This agreement is highly non-trivial, as it involves the difference between expectation values of two, generally non-commuting, Hamiltonians evaluated in distinct states, highlighting the physical consistency of the EPM protocol.

Defining a proper distribution of change in energy $P(\Delta E)$ is essential because fluctuation theorems, such as the Crooks and Jarzynski equalities, are fundamental statements about the statistics of energy (or work) and its fluctuations, not merely its average. These relations connect microscopic reversibility to macroscopic irreversibility through the full probability distribution of energy changes, thereby providing a statistical foundation for the second law of thermodynamics. A well-defined distribution $P(\Delta E)$ is thus the cornerstone for extending fluctuation theorems to genuinely quantum regimes, where coherence and {entanglement} can no longer be ignored.

{Information about the statistics of energy fluctuations is conveniently encoded in the characteristic function, $\mathcal{G}(u)$, defined generically as the Fourier transform of energy-change distribution $P (\D E)$
\begin{equation}
    \mathcal{G}(u) := \langle e^{i u \D E} \rangle_{P (\D E)} = \int d \D E e^{i u \D E}  P (\D E)
\end{equation} 
As in classical probability theory, $\mathcal{G}(u)$ compactly stores the full information about the underlying distribution $P (\D E)$: its Fourier inversion yields $P (\D E)$, and its derivatives at $u=0$ yield the moments and cumulants of $\D E$. These features are completely protocol independent.

For the EPM protocol, using the statistical independence of virtual initial and final projective energies, the characteristic function takes the specific factorized form~\cite{PhysRevA.104.L050203}
\begin{equation}\label{eq:EPM-characteristic}
   \mathcal{G}(u) =  \tr \left(\rhi e^{- i u \Hi} \right)  \tr \left( \Phi[\rhi] e^{ i u \Hf}\right)
\end{equation}
%
This factorization is a distinctive structural feature of EPM and will play an important role in analyzing how initial coherence and entanglement modify fluctuation relations.

Beyond being a compact encoding of the distribution, the characteristic function naturally interfaces with fluctuation theorems. A standard route to Crooks and Jarzynski fluctuation theorems is to study symmetry properties of the characteristic function. For processes beginning in a Gibbs state {(at inverse temperature $\B$)}, the forward and backward characteristic functions in the standard TPM setting satisfy the {symmetry} relation~\footnote{The shift $u \rightarrow u + i \beta$ appearing in Eq.~\eqref{eq:Char-symm} is formally analogous to the Kubo–Martin–Schwinger (KMS) condition for thermal equilibrium states, which states that the thermal correlations functions \SuM{of any two operators $A$ and $B$} satisfy $\langle A(t) B \rangle = \langle A(t + i \B) B \rangle  $ encoding analyticity in a strip of width $\B$ in imaginary time. This analogy is purely structural—the KMS shift encodes equilibrium imaginary-time periodicity, while the Crooks symmetry encodes microscopic reversibility—but it provides useful intuition for understanding why fluctuation theorems emerge from analytic properties of $\mathcal{G}(u)$. For background on the KMS condition, see Refs.~\cite{Haag:1967sg, 10.1143/PTPS.64.12}}
\begin{equation}\label{eq:Char-symm}
    G_F(u) = G_B(-u + i \beta),
\end{equation}
as shown in Refs.~\cite{talkner2007, PhysRevE.75.050102}. Fourier transforming this symmetry relation yields Crooks FT, while integrating it gives Jarzynski equality.}

\subsection{Coherence-corrected Jarzynski equality and entropy-production fluctuation theorem}\label{sec:EPM-rev::subsec:FT}

In order to isolate the correction to the Jarzynski equality arising due to the coherence in the initial state, the authors of Ref.~\cite{PhysRevA.104.L050203} parameterized the initial coherent state as $\rhi = \gbi + \chi$
, where $\gbi$ denotes a thermal state at inverse temperature $\B$ with respect to the initial Hamiltonian $\Hi$ and $\chi$ contains only off-diagonal terms in the eigenbasis of $\Hi$. {We refer to $\chi$ as the coherence operator.} The resulting Jarzynski equality is 

\begin{align}\label{eq:Jarzynski}
    \langle e^{-\B (\D E - \D F)} \rangle =  d \Big[ \tr \left( \gbf \Phi [\gbi] \right)  +  \tr \left( \gbf \Phi [\chi] \right) \Big]
\end{align}
$\gbf$ is the thermal state of the final Hamiltonian $\Hf$ at inverse temperature $\B$, with $\D F = - \B^{-1} \ln (\zbf/\zbi)$ being the free energy difference between $\gbi$ and $\gbf$. The partition functions of $\gbi$ and $\gbf$ are $\zbi = \sum_l e^{-\B \Eil}$ and $\zbf = \sum_k e^{-\B \Efk}$, respectively. The Jarzynski equality for the  TPM scheme is $\langle e^{-\B (\D E - \D F)} \rangle =1$~\cite{Rastegin:2013plm}. The deviation of the right-hand side in~\eqref{eq:Jarzynski} from unity is due to the EPM protocol. The first term on the right hand side, $\tr \left( \gbf \Phi [\gbi] \right) $ is due to uncertainty that results from absence of initial state measurement. The second term $ \tr \left( \gbf \Phi [\chi] \right)$ quantifies the deviation due to initial quantum coherence.

The Crooks-type detailed FT for entropy production in the presence of initial state coherence was derived for the EPM scheme using Crook's formalism~\cite{PhysRevE.60.2721} in Ref.~\cite{Hernandez-Gomez:2022xor}. Let the action of the CPTP map $\Phi$ of the forward process on quantum states be represented as
\begin{align}
    \Phi [\bullet] = \sum_\alpha A_\alpha \bullet A_\alpha^\dagger
\end{align}
where $A_\alpha$ are the Kraus operators satisfying $\sum_\alpha A_\alpha^\dagger A_\alpha = \mathbb{I}$. If the map admits a non-singular fixed point $\pi$, such that $\pi = \Phi[\pi]$, one can define a backward (or time-reversed) map $\tilde{\Phi}$~\cite{PhysRevA.77.034101} as
\begin{align}\label{eq:Phi-dual-Kraus}
    \tilde{\Phi}[\bullet]  = \sum_\alpha \tilde{A}_\alpha \bullet \tilde{A}_\alpha^\dagger, \quad \tilde{A} =  \pi^{1/2} A_\alpha^\dagger \pi^{- 1/2}
\end{align}
The backward process is then characterized by the initial state $\rhti$ and the map $\tilde{\Phi}$. The initial states of the forward and backward processes are decomposed as
\begin{align}
    \rhi & = \gbi + \chi \\
    \rhti &=\gbf + \chit
\end{align}
The terms $\chi$ and $\chit$ contain the coherences of the corresponding states in the eigenbasis of $\Hi$ and $\Hf$, respectively. They do not have any diagonal terms. Within the EPM framework, the joint energy-measurement probabilities for the forward and backward processes are respectively
\begin{align}
\PG (l,k) &=  \pinl \pfink \label{PG-EPM}\\
\PGt (k,l)&=\ptink \ptfinl \label{PGt-EPM}
\end{align}
where $\ptink := \tr (\rhti \Pifk  )$ and $\ptfinl  := \tr (\Phtil[\rhti] \Piil )$ and $\pinl $ and $\pfink$ as defined before. The stochastic entropy production associated with the (virtual) trajectory $| \Eil \rangle \rightarrow | \Efk \rangle$ is then defined as $\D \stotlk :=\ln (\PG (l,k)/\PGt (k,l))$~\cite{RevModPhys.83.771, QST-book-Strasberg,Manzano:2021apn}. The detailed FT for entropy production is
\begin{subequations}\label{eq:DFT-EPM}
\begin{align}
   \D \stotlk &
   =  \B (\D E_{l,k} - \D F) + \D \scorlk \label{eq:DFT-EPM-a} \\
   \D \scorlk &=\D \s_{l,k} + \D \SB_{l,k} \label{eq:DFT-EPM-b}
\end{align}
\end{subequations}
The term $ \B (\D \Elk - \D F)$ is the TPM entropy production~\cite{talkner2007}. {The explicit expressions for $\D\s_{l,k}$ and $\D \SB_{l,k}$ are given by
\begin{align}\label{eq:sigma_terms}
     \sfk (\gbi) := \ln & \pfink (\gbi), \quad  \stfl(\gbf) := \ln \tilde{p}^f_l(\gbf) \nonumber\\
     \D \slk &:= \s^f_{k} (\gbi) -  \st^f_l(\gbf)
\end{align}

\begin{align}\label{eq:Sigma_terms_Old}
     \SBfk (\chi) := \text{ln} \left[ 1 + \frac{\pfink(\chi) }{\pfink(\gbi)} \right], \  & \SBtfl (\chit) := \text{ln} \left[ 1 + \frac{\ptfinl(\chit)}{\ptfinl(\gbf)} \right] \nonumber\\
     \D \SB_{l,f} := \SBfk (\chi) &- \SBtfl (\chit)
\end{align}
$\D \scorlk$ is a novel contribution to entropy production owing to the EPM protocol. The terms $\D\s_{l,k}$ and $\D \SB_{l,k}$ are contributions due to classical uncertainty and coherence, respectively. In particular, $\D\s_{l,k}$ shows up because of the classical statistical uncertainty introduced by the lack of initial projective energy measurement in the EPM scheme. The term $\D \SB_{l,k}$ isolates entropy production due to initial quantum coherence and it is termed as {coherence-affected irreversible entropy production}~\cite{Hernandez-Gomez:2022xor}.}

%
Integrating the detailed FT, Eq.~\eqref{eq:DFT-EPM}, with respect to the forward probability distribution $P_\G$ gives the integral FT for the average entropy production
\begin{align}\label{eq:IFT-EPM}
    \Big\langle  e^{-\D \stot} \Big\rangle_\G = 1,
\end{align}
where $\langle \bullet \rangle_\G$ denotes average with respect to the forward probability distribution $P_\G$. The application of Jensen's inequality to integral FT Eq.~\eqref{eq:IFT-EPM} yields $\langle \D \stot \rangle_\Gamma  \geq 0$, which is the statement {of the second law of thermodynamics in this context,} that the average entropy production is non-negative. 

Additionally, $\D \SB$ satisfies its own integral FT
\begin{equation}\label{eq:IFT-EPM-SB}
    \langle e^{- \D \SB} \rangle_\G = 1 \quad \Rightarrow \langle \SB \rangle_\G \geq 0
\end{equation}
Recall that Refs.~\cite{PhysRevA.104.L050203, Hernandez-Gomez:2022xor} decomposed initial coherent state as $\rhi = \gbi + \chi$ where $\gbi$ is the thermal density matrix of $\Hi$ at inverse temperature $\B$ and $\chi$  denotes the coherent part of $\rhi$ (the off-diagonal component in $\Hi$ basis). By construction, the coherence operator $\chi$ is not a physical state: it need not be positive semi-definite, has $\tr \chi = 0$, and cannot be prepared as a standalone density matrix. Since both $\rhi$ and $\gbi$ can be prepared experimentally, in the experiments, the terms in the FTs involving $\chi$ are obtained by expressing $\chi$ as $\chi = \rhi - \gbi$ and using the linearity of trace and CPTP maps. 

Furthermore, this decomposition is only valid for initial coherent states whose populations in the eigenbasis of $\Hi$ are the same as those of $\gbi$. This is a highly restrictive assumption and will fail to hold for most initial coherent states.

To overcome these drawbacks, we propose an alternative that quantifies the effect of coherence and entanglement on energy-fluctuation statistics using only valid, experimentally preparable quantum states. Instead of introducing an auxiliary, unphysical operator $\chi$, we construct state-only, {resource-theoretic decompositions} of $\rhi$. These decompositions express $\rhi$ as a convex combination of a thermal state and a minimal state that carries a well-defined amount of the relevant resource-athermality, coherence, or entanglement. Because every component appearing in the {resource-theoretic decomposition} of $\rhi$ is itself a legitimate density matrix, all the quantities entering the corrected FTs-whether probability ratios in the detailed FT or trace functionals in the Jarzynski equality- are experimentally accessible functions of physical states. In practice, one can prepare each component state and implement the EPM protocol on them individually to obtain the required correction terms, ensuring that the resulting fluctuation relations remain fully operational and experimentally testable.

This state-only approach keeps the analysis experimentally accessible while removing ambiguities associated with $\chi$. It also ensures that each contribution classical, athermal, coherent, or entangled is unambiguously isolated and can be interpreted in {resource-theoretic} terms.


\section{Improved EPM Fluctuation Theorems}\label{sec:EPM-FT}
 
In this section, we derive the improved versions of both the Jarzynski equality and the detailed entropy-production
FT within the EPM framework, valid for general initial states that may exhibit athermality and coherence. 
We replace the ad hoc decomposition $\rho_i=\gamma_{\beta,i}+\chi$
introduced in Refs.~\cite{PhysRevA.104.L050203, Hernandez-Gomez:2022xor} with {resource-theoretic} decompositions based
exclusively on experimentally preparable states. 
For this, we consider two quantities called the {weight of
athermality} (see Def.~\ref{def:athermality} below) and {weight of coherence} (see Def.~\ref{def:coherence} below), which allow us to {uniquely and operationally} isolate the contributions of athermality and coherence to non-equilibrium FTs. 
 
\begin{definition}[\emph{Weight of athermality}]\label{def:athermality}
The weight of athermality of a quantum state $\rho$
with respect to a thermal state $\gamma_{\beta,H}$ is defined as
\[
A_w(\rho)
   = \min_{\tau \in \mathscr{D}(\mathcal{H})}
      \{ a \ge 0 : \rho = (1-a)\gamma_{\beta,H} + a\tau \}.
\]
Here, the thermal state, $\gamma_{\beta,H}$, is defined with respect to the Hamiltonian $H$ at a finite non-zero inverse-temperature $\B$. 
While $\tau$ is chosen from $\mathscr{D}(\mathcal{H})$, which is the set of all density matrices supported on the Hilbert space $\mathcal{H}$.
\end{definition}
By construction, $0 \leq a \leq 1$, and the minimal athermal state $\tau$ captures precisely the excess non-equilibrium content of $\rho$ beyond the thermal baseline $\gamma_{\beta,H}$.
The weight of athermality $a$ and the corresponding minimal athermal state $\tau$ are uniquely defined as
\begin{subequations}
\begin{align}
    a &= 1 - \mu_{\min} \left( (\gamma_{\beta,H})^{\frac{1}{2} }\ \rho \ (\gamma_{\beta,H})^{-\frac{1}{2}} \right), \label{eq:weight_atherm_solns-a}\\
    \tau &= \frac{\rho - (1 - a) \gamma_{\beta,H}}{a}, \label{eq:weight_atherm_solns-b}
\end{align}\label{eq:weight_atherm_solns}
\end{subequations}
where ${\mu}_{\min}(\mathcal{A})$ denotes minimum eigenvalue of Hermitian operator $\mathcal{A}$. The existence and uniqueness of $a$ and $\tau$ is ensured by the full rank of $\gamma_{\beta,H}$. 
A short proof is included in the Appendix~\ref{app:weight_atherm}.

\begin{definition}[\emph{Weight of coherence}~\cite{PhysRevA.97.032342, PhysRevA.102.032406}]\label{def:coherence}
Let $\mathscr{I}$ be the set of states diagonal in a fixed reference basis $\mathfrak{B}$. The weight of coherence (with respect to the basis $\mathfrak{B}$) of a coherent quantum state $\rho$ is defined as
\[
C_w(\rho) 
= \min_{\substack{\{ \sigma, \tau \} }} \{ c \geq 0 : \rho = (1 - c)\sigma + c \tau, \sigma \in \mathscr{I}, \tau \in \mathscr{D}(\mathcal{H}) \}.
\]
The state $\tau$ is optimally chosen from $\mathscr{D}(\mathcal{H})$ and has coherence with respect to $\mathfrak{B}$. 
\end{definition}
%

\textit{Remark.} The weight of coherence introduced above is not an isolated construction, but rather a special case of a broader class of weight-based resource quantifiers defined via Best Free Approximations (BFA) in general convex quantum resource theories~\cite{PhysRevA.102.032406}. In this framework, a given quantum state is decomposed into a convex combination of a free state and a minimal resourceful state, with the corresponding weight quantifying the amount of the resource present. From this perspective, the weight of coherence corresponds to the BFA associated with the set of incoherent states.
Consequently, the fluctuation theorems derived below for isolating the impact of coherence on non-equilibrium energy and entropy fluctuations admit a natural and systematic generalization. By replacing the incoherent free set with any other choice of free states and employing the corresponding BFA, one can construct fluctuation theorems that isolate the thermodynamic contribution of arbitrary quantum resources within the EPM framework. This highlights the generality of our approach and establishes EPM-based fluctuation relations as a versatile operational probe of quantum resources in non-equilibrium thermodynamics.



\subsection{Jarzynski equality}
To derive the Jarzynski equality in the presence of initial coherence, we first express the initial state as a convex combination of a thermal state and a minimal athermal state using the weight of athermality,
\begin{align}\label{eq:coh-atherm-decomp}
    \rhi &= (1 - a)\gbi + a \tau 
\end{align}
where $a$ is the {weight of athermality} of $\rhi$ with respect to $\gbi$
and $\tau$ is the corresponding minimal athermal state. We assume that  $\gbi$ is full-rank.
Substituting this decomposition into the EPM characteristic function, Eq.~\eqref{eq:EPM-characteristic}, allows us to separate the equilibrium and athermal contributions to the Jarzynski equality. The EPM Jarzynski equality reads
\begin{align}\label{eq:Jarzynski-atherm}
    &\Big\langle e^{-\B (\D E - \D F)}\Big \rangle = \nonumber \\ & \left\{ (1-a) d + a  \tr ( (\gbi)^{-1} \tau ) \right\}  \nonumber \\ 
      &\left \{ (1-a) \tr ( \gbf \Phi[\gbi] ) 
    + a \tr (\gbf \Phi [\tau] ) \right \}
\end{align}
The term $d(1-a)^2\tr ( \gbf \Phi[\gbi] )$ is the thermal contribution, while the rest of the terms are due to athermality.
{For the initial thermal state, we set $a = 0$, and obtain the same expression as in the original EPM protocol of Refs.~\cite{PhysRevA.104.L050203, Hernandez-Gomez:2022xor} (discussed in Sec.~\ref{sec:EPM-review}) decomposition with no coherence in the initial state, viz. Eq.~\eqref{eq:Jarzynski} with $\chi = 0$. 
}

We then further decompose $\tau$ via the {weight of coherence} (with respect to the eigenbasis of $\Hi$) into an incoherent state and a minimally coherent state,
\begin{align}\label{eq:tau-coh-decomp}
    \tau &= (1 -c ) \taud + c \tauc
\end{align}
where $\taud$ is diagonal in the $\Hi$ eigenbasis and $\tauc$ contains the minimal amount of coherence compatible with $\tau$.

Substituting Eqs.~\eqref{eq:coh-atherm-decomp} and \eqref{eq:tau-coh-decomp} into the expression for EPM characteristic function, Eq.~\eqref{eq:EPM-characteristic}, 
the EPM Jarzynski equality becomes
\begin{align}\label{eq:Jarzynski-coh}
    &\langle e^{-\B \left( \D E  - \D F\right)} \rangle = \nonumber \\
    &\Big \{ (1 - a)d + a (1 - c)  \tr((\gbi)^{-1} \taud)  + a c  \tr ((\gbi)^{-1}  \tauc)\Big \} \nonumber \\
    &\Big \{  (1 - a) \tr (\gbf \Phi[\gbi])  + a (1 - c) \tr (\gbf \Phi[\taud]) \nonumber \\
    &+ a c \tr (\gbf \Phi[\tauc]) \Big \}
\end{align}
{The terms 
$  \tr((\gbi)^{-1} \taud)$ and $ \tr (\gbf \Phi[\taud])$ quantify the corrections due to diagonal athermality encoded in $\taud$, while the terms 
$  \tr ((\gbi)^{-1}  \tauc)$ and $ \tr (\gbf \Phi[\tauc]) $ isolate the contribution of genuine coherence (encoded in $\tauc$).} When $c=0$, the expression reduces to the athermality-corrected Jarzynski equality in Eq.~\eqref{eq:Jarzynski-atherm}. 
\subsection{Entropy production}
We now derive the detailed entropy-production FT in the presence of coherence. For the forward and backward processes, we denote the forward and backward initial states as $\rhi$ and $\rhti$, respectively. They can be decomposed into their respective athermal and coherence contributions in the following way, 
\begin{subequations}
    \begin{align}
        \rhi &= (1 - a)\gbi + a (1-c) \taud + a c \tau_c  \label{state-convexMix-fwd}\\
        \rhti &= (1 - \atl)\gbf + \atl (1-\ctl) \tautd + \atl \ctl \tautc \label{state-convexMix-bkd}
    \end{align}
    \end{subequations}
    where $(a,c,\tau_d,\tau_c)$ and $(\tilde{a},\tilde{c},\tilde{\tau}_d,\tilde{\tau}_c)$ are the athermality and coherence parameters and
corresponding minimal states of the forward and backward initial states, respectively. Throughout this section, we assume that the thermal states $\gamma_{\beta,i}$ and $\gamma_{\beta,f}$ are full-rank.
Before deriving the FT, let us define 
\begin{align*}
    \pinl(\mathcal{A}) &\coloneq \tr(\mathcal{A} \; \Pi^i_l ), \quad 
    \pfink(\mathcal{A}) \coloneq \tr(\Phi(\mathcal{A}) \; \Pifk ), \\
   \ptink(\mathcal{A}) &\coloneq \tr(\mathcal{A} \; \Pifk ), \quad
   \ptfinl (\mathcal{A}) \coloneq \tr \left( \Phtil (\mathcal{A}) \; \Piil   \right),
\end{align*}
with $\mathcal{A}$ a generic linear operator.

As seen from Eq.~\eqref{PG-EPM}, the probability distribution of the forward trajectory is $\PG =  p_l^i(\rho_i) \pfink(\rho_i)$. Let us look at the structure of these two probability distributions.
\begin{align}
  p_l^i(\rho_i) 
    &= p_l^i(\gamma_{\B,i}) e^{\Thil(\taud)}e^{\SBil(\tauc)} 
\end{align}
Similarly, we have 
\begin{align}
     \pfink(\rho_i) = p_k^f(\gamma_{\B,i}) e^{\Thfk(\taud)}e^{\SBfk(\tauc)}.
\end{align}
{The terms $p_l^i(\gamma_{\B,i})$ and $p_k^f(\gamma_{\B,i})$ are the thermal contributions to the forward trajectory.  The terms $e^{\Thil(\taud)}$ and $e^{\Thfk(\taud)}$ are the contributions due to \SuM{diagonal} athermalilty, whereas the terms $e^{\SBil(\tauc)} $ and $e^{\SBfk(\tauc)}$ are due to coherence.} The explicit expressions of $\Thil(\taud)$, $\Thfk(\taud)$, $\SBil(\tauc)$ and $\SBfk(\taud)$ 
are given in Appendix~\ref{app:Crooks:subsec:coh}.

Thus, the probability distribution of the forward trajectory comes out to be,
\begin{align}
\label{For-prob}
    \PG (l,k) &=  \pinl(\rho_i) \pfink(\rho_i) \nonumber \\
    &= p_l^i(\gamma_{\B,i}) p_k^f(\gamma_{\B,i}) e^{\Theta_{l,k}^\G(\taud)}e^{ \Sigma_{l, k}^\G(\tauc)}.
\end{align}
Similarly, for the reverse trajectory with initial state $\tilde \rho_i$ as given in Eq.~\eqref{state-convexMix-bkd}, we have
\begin{align}
\label{Rev-prob}
     \SuM{\PGt (k,l)} &= \tilde p_k^i(\gamma_{\B,f}) \tilde p_l^f(\gamma_{\B,f}) e^{\Theta_{l, k}^{\Gt}(\tautd)}e^{ \Sigma_{l,k}^{\Gt}(\tautc)},
\end{align}
where we have defined the new quantities \SuM{$\mathcal{M}_{l, k}^\G$ and $\mathcal{M}_{l, k}^{\Gt}$}, with $\mathcal{M} = \{\Theta, \Sigma\}$, by collecting the contributions arising from both the initial virtual measurement and the final projective measurement, as follows,
\SuM{
\begin{align*}
    \mathcal{M}_{l,k}^\G = \mathcal{M}_l^{(i)} + \mathcal{M}_k^{(f)}\\
     \mathcal{M}_{l,k}^{\Gt} = \tilde{\mathcal{M}}_k^{(i)} + \tilde{\mathcal{M}}_l^{(f)}.
\end{align*}
}
The explicit expressions of 
$\Thtfl(\tautd)$, $\Thtik(\tautd)$, $\SBtfl(\tautc)$ and $\SBtik(\tautc)$ are given in Appendix~\ref{app:Crooks:subsec:coh}.

\SuM{The trajectory-level entropy-production is defined as} $\D s_\text{tot}^{l,k} := \ln (\PG(l,k)/\SuM{\PGt(k,l)})$. From Eqs.~\eqref{For-prob}~and~\eqref{Rev-prob}, we get the  detailed entropy-production \SuM{FT} as,
\begin{equation}
\boxed{
\begin{aligned}
\D \stotlk &= \beta\,(\Delta \Elk - \Delta F) + \D \scorlk, \\[4pt]
\D \scorlk &= \D \slk + \D \Thlk + \D \SBlk 
\end{aligned}
}
\label{eq:DFT-entropy-single-party}
\end{equation}
where 
\SuM{$\D \scorlk$ is the correction term due to athermality and coherence on top of the usual expression derived via TPM}~\cite{talkner2007} . The correction term is composed of three contributions - $\D \slk $ \SuM{(defined earlier in Eq.~\eqref{eq:sigma_terms})} captures classical uncertainty (no initial projective measurement), $\D \Thlk$ isolates athermality, and $\D \SBlk$ isolates coherence. \SuM{The explicit expressions of $\D \Thlk$ and $\D \SBlk $ are provided in App.~\ref{app:Crooks:subsec:coh}.}
In general, each of $\D \Thlk$ and $\D \SBlk $ consists of two pieces,
\begin{subequations}\label{eq:Theta-Sigma-Final}
\begin{align}
    \D \Thlk &=  \Thlk^\G(\taud) -  \Theta_{l,k}^{\Gt}(\tautd) \label{eq:Theta-Final} \\
    \D \SBlk &=  \SBlk^\G(\tauc) -  \Sigma_{l,k}^{\Gt}(\tautc) \label{eq:Sigma-Final} 
\end{align}
\end{subequations}

The integral entropy-production FT follows from the detailed FT, Eq.~\eqref{eq:DFT-entropy-single-party}, by averaging over the forward probability distribution $P_\G$, and takes the same form as Eq.~\eqref{eq:IFT-EPM}, viz. $\langle e^{- \D \stot}\rangle_\G = 1$. Application of Jensen's inequality gives $\langle \D \stot\rangle_\G \geq 0$. By substituting $\D \stot$ from Eq.~\eqref{eq:DFT-entropy-single-party} we get
\begin{equation}
    \B \left(\langle \D E \rangle - \D F \right) + \left( \langle \D \s \rangle  + \langle \D \Th \rangle + \langle \D \SB \rangle  \right)  \geq 0
    \label{eq:II-law-EPM-coh}
\end{equation}
%
Recall that the quantity $\D \stot^{\rm TPM}:=\B \left(\langle \D E \rangle - \D F \right)$ is the standard TPM average entropy production. Subsequently, Eq.~\eqref{eq:II-law-EPM-coh} can be written as
\begin{equation}
    \D \stot^{\rm TPM} + \left( \langle \D \s \rangle  + \langle \D \Th \rangle + \langle \D \SB \rangle  \right)  \geq 0
\end{equation}
As a result, Eq.~\eqref{eq:II-law-EPM-coh} is a 
{refined} statement of the 
{second law of thermodynamics}~\cite{Esposito2010, PhysRevE.99.012120, jiang2018improved,  Bera2017, PhysRevA.110.012451, mondal2023modified, Aimet2025} in this context.

When $c = \tilde{c} = 0$
the coherence corrections vanish, $\Delta\Sigma_{l,k} = 0$, and we recover the athermality-only \SuM{entropy production}, 
whereas
for $a = \tilde{a} = 0$ both $\Delta\Theta_{l,k}$ and $\Delta\Sigma_{l,k}$ vanish and the standard TPM entropy production is
retrieved. This completes the characterization of coherence-induced modifications to the EPM FTs. The derivations of coherence-corrected FTs are given in Apps.~\eqref{app:Jarzynski:subsec:coh} and~\eqref{app:Crooks:subsec:coh}.

\section{EPM Fluctuation Theorems for bipartite quantum systems}\label{sec:EPM-FT-biparty}

Non-equilibrium thermodynamic processes often involve multipartite quantum systems in which both classical correlations and entanglement can influence the statistics of energy changes. The EPM protocol discussed in Sec.~\ref{sec:EPM-rev:subsec:proto} employs a {global} projective energy measurement at the final time~$\tf$, and therefore cannot reveal how entanglement between subsystems contributes to FTs. To capture such effects, the protocol must be generalized to the bipartite setting while retaining the defining features of the single-system EPM scheme namely, the absence of an initial projective measurement, the virtual assignment of initial energies, and projective energy measurements performed only at the final time, now carried out {locally} on each subsystem. 

\subsection{EPM Protocol for bipartite quantum systems}\label{sec:EPM-FT-biparty::subsec:biparty-protocol}

To generalize the EPM protocol to bipartite systems without altering its measurement structure (virtual initial energies, no initial collapse, final projective measurements only), we allow the two subsystems to interact only during the time interval $(\ti,\tf)$ while performing both the virtual initial energy assignments and the final projective energy measurements {locally}. Concretely, the system begins at $\ti$ with a non-interacting Hamiltonian $\HAi + \HBi$, ensuring that the virtual initial energies $(\EilA, \EilB)$ 
are sampled with respect to the local eigen-projectors of $\HAi$ and $\HBi$. 
The interaction $\Hint(t)$  is then switched on during the evolution governed by a global CPTP map $\Phi$, and is switched off again at $\tf$, at which point independent local projective energy measurements in the eigenbases of $\HAf$ and $\HBf$ are performed. Switching off the interaction at $\tf$ ensures that the final energy is measured with respect to local Hamiltonians $\HAf$ and $\HBf$, preserving compatibility with the EPM measurement structure. This construction preserves the measurement structure of the original EPM scheme,
while allowing initial correlations, including entanglement, to influence the joint statistics through the dynamics.

With these local measurements, the EPM joint distribution becomes
\begin{align}\label{eq:pcoh}
    \pent^{\lb,\kb} &\coloneqq p(\EilA, \EilB; \EfkA, \EfkB) \nonumber \\ &= \pinlb(\EilA, \EilB) \pfinkb( \EfkA, \EfkB)
\end{align}
where $\pinlb := \tr(\rhi \Pii_{\lA} \otimes \Pii_{\lB})$ and $\pfinkb:= \tr(\Phi[\rhi] \PifkA \otimes \PifkB)$ are the probabilities associated with the initial virtual and final local energy measurements and $\pent^{\lb,\kb}$ is joint probability distribution satisfying $\sum_{\lb , \kb} \pent^{\lb , \kb} = 1$; $\lb \equiv \{ \lA, \lB \}$, $\kb \equiv \{ \kA, \kB  \}$. As in the single-system EPM scheme, this is not a genuine joint measurement probability but the product of statistically independent initial virtual and final projective statistics. $\PiilA := |\EilA \rangle \langle \EilA |$, $\PiilB := | \EilB \rangle \langle \EilB |$ are the projectors onto the eigenstates of the initial local Hamiltonians $\HAi$ and $\HBi$, respectively. Similarly, the projectors $\PifkA$ and $\PifkB$ are defined that project onto the eigenstates of the final local Hamiltonians. 
These probabilities define the bipartite EPM distribution of energy changes,
\begin{align}
    \Pe(\D E)  \coloneq  \sum_{\lb , \kb}&  \pent^{\lb,\kb}  \delta(\D E - \D E_{\lb,\kb} )  
\end{align}
where $\D E_{\lb,\kb} = \EfkA + \EfkB - \EilA - \EilB$.
Because the initial virtual energies are sampled locally but the evolution is global, this distribution carries signatures of classical correlations and entanglement present in the initial state. In this section, we isolate the specific contribution of entanglement to these energy-change statistics and to the resulting FTs.


We now present two different 
and operationally motivated ways of decomposing the initial bipartite state, each leading to a distinct family of entanglement-corrected FTs. The expression for \SuM{trajectory-level} entropy production in both cases takes the following unified form

\begin{equation}
\boxed{
\begin{aligned}
\D \stotlbkb &= \beta\,(\Delta \Elbkb - \Delta F) + \D \scorlbkb, \\[4pt]
\D \scorlbkb &= \D \slbkb + \D  \corr_{\textbf{l},\textbf{k}}
\end{aligned}
}
\label{eq:DFT-entropy-bi-party}
\end{equation}
where $\D \slbkb $ captures classical uncertainty (no initial projective measurement) as before, whereas $\D \corr_{\textbf{l},\textbf{k}}$ captures the effect of total correlations. In the following subsections, we evaluate the contributions from correlations explicitly.

\subsection{Fluctuation Theorems: decomposition I}

We first consider an operationally simple way of separating local thermodynamics from nonlocal correlations, which we refer to as decomposition~I.
%
%
In this approach, the initial bipartite state $\rhi$ is written as
\begin{equation}\label{eq:ent-ini-forward-ver1}
    \rhi = \rhAi \otimes \rhBi + \EnAB,
\end{equation}
where 
\begin{subequations}
\begin{align}
    \rhAi &= \gbAi := e^{-\B \HAi}/\zbAi, \\
    \rhBi &=\gbBi :=  e^{-\B \HBi}/\zbBi,
\end{align}
\end{subequations}
are local thermal states of $\HAi$ and $\HBi$ at inverse temperature $\B$, with partition functions $\zbAi$ and $\zbBi$, respectively. 
The operator $\EnAB$ contains all nonlocal correlations and is defined such that
\begin{equation}
    \tr_A \left( \EnAB\right) = \tr_B \left( \EnAB\right) = 0
\end{equation}
Thus, even when $\rhi$ is entangled, its local marginals remain thermal. Note, however, that $\EnAB$ is not a density matrix 
and may include both classical and quantum correlations. It cannot be prepared as a standalone quantum state. We refer to $\EnAB$ as the correlation operator.

Using Eq.~\eqref{eq:ent-ini-forward-ver1} in the EPM characteristic function and following the same steps as in Sec.~\eqref{sec:EPM-FT}, the EPM Jarzynski equality becomes
\begin{align}\label{eq:Jarzynski-EnAB-biparty}
    &\langle e^{-\B (\D E  - \D F)} \rangle = \nonumber \\
    &\left \{d +  \tr ( (\gbi)^{-1} \EnAB) \right \} \nonumber \\
    &\Big \{ \tr(\gbf \Phi[\gbi]) 
     + \tr(\gbf \Phi[\EnAB] ) \Big \}
\end{align}
where 
\begin{align}
    \gbi = \gbAi \otimes \gbBi, \quad \text{and} \quad \gbf = \gbAf \otimes \gbBf.
\end{align}
The terms other than $d \, \tr(\gbf \Phi[\gbi]) $ that involve $\EnAB$ originate from non-local correlations in $\rhi$.
%
%
For the detailed entropy-production FT, we decompose the initial states
of the forward and backward processes as

\begin{align}
    \rhi &= \gbAi \otimes \gbBi + \EnAB \\
    \rhti &= \gbAf \otimes \gbBf + \EntAB 
\end{align}
and derive detailed FT for entropy production in the unified FT structure of Eq.~\eqref{eq:DFT-entropy-bi-party}. The derivation is presented in App.~\ref{app:Crooks:subsec:biparty-EnAB}. The classical uncertainty contributions are

\begin{subequations}\label{eq:classical-unceratain-detailedFT-biparty}
    \begin{align}
        \sfkb &:= \ln \pfinkb(\gbi) \\
        \stflb &:= \ln \ptfinlb (\gbf)
    \end{align}
\end{subequations}
with $\D \slbkb := \sfkb - \stflb$. 

The contribution of correlation, as discussed in Eq.~\eqref{eq:DFT-entropy-bi-party} turns out to be
\begin{equation}
\label{eq:DeltaPsi-residual}
    {\D \corr_{\textbf{l},\textbf{k}} :=  \Psi^{\G}_{\lb, \kb}(\EnAB) -  \Psi^{\Gt}_{\lb, \kb}(\EntAB).}
\end{equation}
Here we have \SuM{$\Psi^{\G}_{\lb, \kb} \coloneqq \Psiilb + \Psifkb$ and  $\Psi^{\Gt}_{\lb, \kb} \coloneqq \Psitikb + \Psitflb$.}
The explicit expressions of these terms are as follows,
 {
\begin{align}
\Psiilb(\EnAB) &\coloneqq   \ln\left( 1 +   \frac{\pinlb(\EnAB) }{\pinlb(\gbi) } \right) \nonumber\\
\Psitikb(\EntAB) &\coloneqq  \ln\left( 1 +   \frac{\ptinkb (\EntAB)}{ \ptinkb(\gbf) }  \right) \nonumber\\
\Psifkb(\EnAB) &\coloneqq  \ln\left( 1 +  \frac{\pfinkb (\EnAB)}{\pfinkb (\gbi)}  \right) \nonumber\\
\Psitflb(\EntAB) &\coloneqq  \ln\left( 1  +  \frac{\ptfinlb (\EntAB)}{\ptfinlb(\gbf)} \right) 
 \end{align}
 }
Although algebraically simple, decomposition~I has significant conceptual limitations:

\begin{enumerate}

\item \emph{$\EnAB$ is not a physical state.} It is neither positive semidefinite nor normalized and cannot be prepared experimentally.
As in the $\chi$-decomposition for coherence in Sec.~\ref{sec:EPM-review}, terms involving $\EnAB$ must be reconstructed indirectly \SuM{from $\rhi$ and $\gbi$} using linearity of the trace and of the CPTP maps.

\item  \emph{Not entanglement-specific.}  $\EnAB$ may contain both classical and quantum correlations. Therefore, the correction $\D \Xilbkb$ does not quantify entanglement alone.

\item  \emph{Not universally applicable.}  A general bipartite state does not always admit a decomposition of the form $\rhi = \rhAi \otimes \rhBi + \EnAB$, particularly when its marginals are not thermal.

\end{enumerate}

These limitations motivate a state-only, resource-theoretic approach that isolates entanglement in a physically meaningful way.

We now develop such an approach using the Best Separable Approximation (BSA).

\subsection{Fluctuation Theorems: decomposition II}\label{sec:EPM-FT-biparty::subsec:decomp-II-BSA}

To overcome the drawbacks of decomposition I, we adopt a fully operational and resource-theoretic decomposition based on the {Best Separable Approximation} (BSA).
\begin{definition}[\emph{Best Separable Approximation}~\cite{PhysRevLett.80.2261}]\label{def:BSA}For any bipartite entangled state $\rho$, the BSA expresses it as 
\begin{equation}\label{eq:BSA}
    \rho =  \lambda \rhE + (1-\lambda) \rhS,
\end{equation}
where $0 \leq \lambda \leq 1$ is the BSA entanglement weight, uniquely defined, $\rhE \in \mathscr{D}(\mathcal{H})$ is an entangled state (the minimal entangled
component), and $\rhS$ is a separable state, expressible as
\begin{equation}\label{eq:BSA-seprable}
    \rhS = \sum_j \rj \rhAj \otimes \rhBj, \quad 0\leq \rj \leq 1,
\end{equation}
where $\rhAj \in \mathscr{D}(\mathcal{H}_A)$ and $\rhBj \in \mathscr{D}(\mathcal{H}_B)$. 
\end{definition}
%
\SuM{The BSA decomposition of Def.~\ref{def:BSA} exists for any bipartite mixed state in arbitrary $M \times N$ dimensions and is unique~\cite{Karnas:2000orh}. The weight of the entangled component, $\lambda$, defines a valid entanglement monotone; in the special case of two-qubit states, $\lambda$ coincides with the concurrence~\cite{PhysRevA.64.052302}. This operational interpretation of $\lambda$ makes the BSA particularly suitable for defining entanglement-induced corrections to fluctuation relations.}

Crucially, every component in Eqs.~\eqref{eq:BSA}-\eqref{eq:BSA-seprable} is a valid,
preparable density matrix, ensuring that all correction terms in the resulting FTs are experimentally accessible.

In order to find the contributions of correlations in the FTs, we express the initial state of the EPM protocol in the BSA form as follows,
\begin{align}\label{eq:ent-ini-forward-BSA}
    \rhi &=  \lambda \rhE + (1-\lambda) \rhS ,
\end{align}
with $\rhS$ given by Eq.~\eqref{eq:BSA-seprable}.
Each local state in $\rhS$ is further decomposed using the weight of athermality and the weight of coherence introduced in Sec.~\ref{sec:EPM-FT}
\begin{align}\label{eq:ent-ini-forward-BSA-local-state-decomp}
    \rhAj &= (1 - \aAj) \gbAi + \aAj (1 - \cAj) \tauAdj + \aAj \cAj \tauAcj, \nonumber 
    \\
     \rhBj &= (1 - \aBj) \gbBi + \aBj (1 - \cBj) \tauBdj + \aBj \cBj \tauBcj 
\end{align}
where $a_j^{A(B)}$, $c_j^{A(B)}$ are, respectively, the weights of
athermality and coherence of $\rho_j^{A(B)}$ with respect to $\gamma_{\beta, i}^{A(B)}$, and $\tau_{d,j}^{A(B)}$ and $\tau_{c,j}^{A(B)}$ are
the corresponding minimal diagonal-athermal and minimal coherent states.
For each $j$ in the sum, $\rhAj \otimes \rhBj$ is a sum of nine terms which we rewrite compactly as 
%
\begin{align}
    \rhAj \otimes \rhBj &=  (1 - \aAj) (1 - \aBj) \gbi + \rhdj + \rhcj, \label{eq:sepBSA-prod}
\end{align}
where $\gbi = \gbAi \otimes \gbBi$ is the global thermal state of initial Hamiltonian $H(t_i) = \HAi + \HBi$ at inverse temperature $\B$, $\rhdj$ collects all diagonal-athermality terms, and $\rhcj$ collects all coherence terms.
The explicit expressions for $\rhdj$ and $\rhcj$ are provided in Eq.~\eqref{eq:BSA-product-decomposition-explicit-expr} of Appendix~\ref{app:Jarzynski:subsec:biparty-BSA}. \SuM{Each of these operators is positive, and the corresponding normalized states are obtained by dividing by their trace, so the entire decomposition is in terms of legitimate density matrices.}

Using the BSA decomposition~\eqref{eq:ent-ini-forward-BSA} together with the local
decompositions~\eqref{eq:sepBSA-prod} in the
EPM characteristic function, the EPM Jarzynski equality becomes
\begin{align}\label{eq:Jarzynski-ent}
    \langle e^{-\beta (\Delta E - \Delta F)} \rangle = \mathcal{J}^{(i)} \mathcal{J}^{(f)},
\end{align}
where
\begin{align} 
     & \mathcal{J}^{(i)} = \nonumber \\
     &  \lambda  \tr \left( \left( \gbi \right)^{-1} \rhE \right) +  (1 - \lambda) \Big \{ \sum_j \rj \Big( (1 - \aAj) (1 - \aBj) d  \nonumber\\
      &+  \tr \left( \left( \gbi \right)^{-1} \rhdj \right)  + \tr  \left( \left( \gbi \right)^{-1} \rhcj \right)  \Big) \Big \}   \nonumber\\
     & \mathcal{J}^{(f)} = \nonumber \\
     & \Big \{ \lambda \tr \left( \gbf \Phi [\rhE] \right) + \nonumber \\
     &(1 - \lambda) \Big \{  \sum_j \rj \Big( (1 - \aAj) (1 - \aBj) \tr \left( \gbf \Phi [\gbi] \right)\nonumber\\ 
     &  + \tr \left( \gbf  \Phi [\rhdj] \right)  + \tr \left( \gbf  \Phi [\rhcj] \right)  \Big)  \Big \} \Big \} \nonumber
\end{align}

The terms involving $\rhE$ constitute an entanglement correction to the Jarzynski equality. The terms involving the sum $\sum_j r_j$ are corrections due to classical correlations. The classical correlation corrections can be further decomposed into \SuM{local} thermal, diagonal athermality, and coherence corrections. In this way, the BSA-based decomposition cleanly disentangles entanglement from local quantum resources and classical correlations.

We now derive the detailed entropy-production FT via BSA decomposition of the initial state. 
For the forward and backward processes, we BSA-decompose the initial states as,
\begin{align}\label{eq:BSA-ini-forward-and-back}
    \rhi &=  \lambda \rhE + (1-\lambda) \rhS , \nonumber \\
   \rhti &=  \lamt \rhtE + (1-\lamt) \rhtS , 
\end{align}
where the separable parts $\rhS$ and $\rhtS$ are decomposed \SuM{further} into the \SuM{local} thermal, athermal, and coherent states as discussed earlier in Eq.~\eqref{eq:BSA-seprable}~and~\eqref{eq:sepBSA-prod}. \SuM{As before,} we represent all the terms corresponding to the backward trajectory with a tilde, in order to distinguish them from the forward trajectory.
We derive the unified detailed entropy production FT of Eq.~\eqref{eq:DFT-entropy-bi-party} for the BSA-decomposed initial states of Eqs.~\eqref{eq:BSA-ini-forward-and-back}. We find that the correction due to correlations can be decomposed into two terms, viz. 
\SuM{
\begin{align}\label{eq:Corre-corrections}
    \D \corr_{\textbf{l},\textbf{k}} &= \D \Lamlbkb + \D \Xilbkb, \\
    \D \Lamlbkb &=  \Lambda^{\G}_{\lb,\kb} - \Lambda^{\Gt}_{\lb,\kb}, \nonumber\\
    \D \Xilbkb &= \Xi^{\G}_{\lb,\kb} - \Xi^{\Gt}_{\lb,\kb}. \nonumber
\end{align}
}
The first term, $\D \Lamlbkb$ isolates classical correlations, whereas the second term $\D \Xilbkb$ isolates entanglement. \SuM{Note that we have,
\begin{subequations}
\begin{align}
    \Lambda^{\G}_{\lb,\kb} &= \Lamilb + \Lamfkb, \quad \Lambda^{\Gt}_{\lb,\kb} = \Lamtikb + \Lamtflb, \\
    \Xi^{\G}_{\lb,\kb} &= \Xiilb + \Xifkb, \quad \Xi^{\Gt}_{\lb,\kb} = \Xitikb + \Xitflb
\end{align}
\end{subequations}
}

Expressions for \SuM{$ \Lamilb$, $\Lamtikb$, $\Xiilb$, and $\Xitikb$} involve the probabilities associated with virtual initial \SuM{local} energy measurements, whereas \SuM{$ \Lamfkb$, $\Lamtflb$, $\Xifkb$, and $\Xitflb$} involve probabilities associated with final \SuM{local} energy measurements. 

Let us first present the contributions due to classical correlations, 
\begin{align}
    \Lamilb &\coloneqq \ln \left( \sum_j \rj \SSilbj \right), \ \
    \Lamtikb \coloneqq  \ln \left(\sum_{\jt} \rtj \SStikbj \right), \nonumber \\
    \Lamfkb &\coloneqq  \ln \left(\sum_j \rj \SSfkbj \right), \ \
    \Lamtflb \coloneqq \ln \left( \sum_{\jt} \rtj \SStflbj \right).
\end{align}
Here $\SSilbj$, $\SStikbj$, $\SSfkbj$, and  $\SStflbj$ are individual contributions of the product states from the decompositions of the maximal separable component $\rhS$, Eq.~\eqref{eq:BSA-seprable}, and simlar one for $\rhtS$
due to the initial virtual and final measurements in the forward and backward trajectories. We present their explicit forms in Appendix~\ref{app:Crooks:subsec:biparty-BSA}. 

Now, let us present the explicit form of terms due to entanglement. 
The explicit expressions of these quantities are as follows,
    \begin{align}
        \Xiilb &:= (1 - \lambda) + \lambda \frac{\pinlb(\rhE)}{\pinlb(\gbi) \Lamilb} \nonumber\\
        \Xitikb &:= (1 - \lamt) + \lamt \frac{\ptinkb(\rhtE)}{\ptinkb(\gbf) \Lamtikb} \nonumber\\
        \Xifkb &:= (1 - \lambda) + \lambda \frac{\pfinkb(\rhE)}{\pfinkb(\gbi) \Lamfkb} \nonumber\\
        \Xitflb &:= (1 - \lamt) + \lamt \frac{\ptfinlb(\rhtE)}{\ptfinlb(\gbf) \Lamtflb}. 
    \end{align}

In summary, decomposition~II provides a fully operational, state-restricted refinement of the EPM FTs that cleanly separates classical uncertainty ($\D \sigma$), classical correlations with local quantum resources ($\D\Lambda$), and entanglement ($\D \Xi$). Unlike decomposition~I, each correction is expressed in terms of experimentally preparable states, allowing direct experimental reconstruction. \SuM{The derivations of FTs presented in this section can be found in Apps.~\ref{app:Jarzynski:subsec:biparty-BSA} and~\ref{app:Crooks:subsec:biparty-BSA}.}

\section{New measures of coherence and entanglement contributions to energy fluctuations}\label{sec:new-measures}

In the previous sections, we derived coherence- and entanglement-corrected fluctuation theorems \SuM{employing resource-theoretic initial-state decompositions}.
Importantly, these decompositions demonstrate that coherence and entanglement can be cleanly separated from classical and thermal contributions {at the level of individual EPM trajectories} $\{l,k\}$ (or $\{\lb,\kb\}$).

While such {trajectory-level} formulations provide detailed insights, many practical applications call for a {single scalar quantity} that quantifies the overall impact of coherence or entanglement on EPM energy fluctuations for a given initial state and driving protocol. In this section, we therefore introduce two operationally meaningful measures: the {coherence fluctuation distance} (CFD)  and the {entanglement fluctuation distance} (EFD). The former quantifies how strongly initial quantum coherence modifies EPM trajectory statistics, while the latter captures the corresponding influence of entanglement in bipartite systems.

Both quantities are defined as Kullback–Leibler (KL) divergences between trajectory probability distributions generated by experimentally preparable initial states. As a result, they inherit a clear operational interpretation and remain fully consistent with the EPM framework.


\subsection{Coherence fluctuation distance}\label{sec:new-measures::subsec-CFD}

To quantify how strongly initial quantum coherence influences the statistics of energy changes in the EPM protocol, we introduce the coherence fluctuation distance (CFD). Let $\mathscr{I}$ denote the set of states diagonal in the eigenbasis of the initial Hamiltonian $\Hi$. For any $\rhi^{\mathcal{I}} \in \mathscr{I}$, the EPM distribution factorizes as
\begin{equation}
    \pcoh^{l,k} = \pinl(\rhi^{\mathcal{I}}) \pfink(\Phi[\rhi^{\mathcal{I}}]),
\end{equation}
and therefore provides a natural incoherent reference against which the statistics generated by a coherent state may be compared.

For a given coherent initial state $\rhi$, a CPTP map $\Phi$ and Hamiltonian $H(t)$, we define the CFD as
\begin{equation}\label{eq:CFT-defn}
    \mathfrak{D}_c (\rhi) := \underset{\rhi^{\mathcal{I}} \in \mathscr{I}}{\min} D_{KL} (\pcoh^{l,k}(\rhi)  || \pcoh^{l,k}(\rhi^{\mathcal{I}}))
\end{equation}
where $D_{KL} (\pcoh^{l,k}(\rhi)  || \pcoh^{l,k}(\rhi^{\mathcal{I}}))$ is the Kullback-Leibler (KL)-divergence between the two probability distributions $\pcoh^{l,k}(\rhi) $ and $\pcoh^{l,k}(\rhi^{\mathcal{I}})$ which are EPM probability distributions associated with initial states $\rhi$ and $\rhi^\mathcal{I}$, respectively. 
The KL-divergence between two probability distribution $p_i$ and $q_i$ is given as,
\begin{align}\label{eq:DKL-pEPM-coh-n-I}
    D_{KL}(p_i\|q_i) = \sum_i p_i \ln\frac{p_i}{q_i}.
\end{align}
%
The minimization in Eq.~\eqref{eq:CFT-defn} runs over {all incoherent states with 
{respect to the eigenbasis} of $\Hi$}.
Thus, for a given process, $\mathfrak{D}_c (\rhi)$ measures the least statistical distinguishability between the EPM distribution generated from $\rhi$ and that of any incoherent state. 
Note that along with the initial coherent state $\rhi$, $\mathfrak{D}_c(\rhi)$ also depends on the CPTP map $\Phi$ and the Hamiltonian $H(t)$, but 
we suppress \SuM{the dependence on the latter two} in our notation for brevity.

 
The 
CFD
quantifies the thermodynamic relevance of coherence, that is if the EPM trajectory distribution produced by $\rhi$ is nearly indistinguishable from that of an incoherent state, then coherence does not {noticeably} modify \SuM{non-equilibrium} energy fluctuations even if it is present in the state. Conversely, large values of CFD certify that coherence \SuM{significantly} biases EPM trajectories relative to their incoherent counterpart, producing measurable deviations in energy change statistics. Note that
\begin{equation}
    \mathfrak{D}_c (\rhi) = 0 \Longleftrightarrow    \pcoh^{l,k} (\rhi) = \pcoh^{l,k}(\rhi^{\mathcal{I}}) , 
\end{equation}
where $\rhi^{\mathcal{I}}$ is some diagonal state.
A vanishing value therefore signifies that coherence in $\rhi$ is thermodynamically silent, that is it does not influence the energy-trajectory statistics under the given dynamics. This can occur when coherence is present but dynamically irrelevant, for instance when neither $H(t)$ nor $\Phi$ couples diagonal and off-diagonal sectors. \SaM{The theorem~\ref{them:1} stated below formalizes this scenario for a single qubit system.} On the other hand, larger values certify that coherence biases non-equilibrium {fluctuations} in such a way that cannot be masked by any incoherent surrogate state.

Evaluating the minimization in Eq.~\eqref{eq:CFT-defn} may not always be analytically tractable. We therefore provide two rigorous upper bounds, one of which we illustrate to be tight and very close to $\mathfrak{D}_c (\rhi)$ itself for single qubit states.
Let $\Delta [\rhi]$ denote the state obtained by acting the dephasing (in the eigenbasis of $\Hi$) channel on the initial coherent state $\rhi$. As a result $\Delta [\rhi]$ is diagonal in the eigenbasis of initial Hamiltonian $H(t_i)$.  We show that
\begin{equation}\label{eq:CFD-bounds}
    \mathfrak{D}_c (\rhi) \leq D_{KL} (\pcoh^{l,k}(\rhi) || \pcoh^{l,k}(\Delta [\rhi])) \leq 2 C_{re}(\rhi)
\end{equation}
%
where $C_{re}(\rhi)$ is the relative entropy of coherence~\cite{PhysRevLett.113.140401}
of state $\rhi$ defined as 
\begin{equation}\label{eq:Cre-defn}
    C_{re}(\rhi) := S(\Delta [\rhi]) - S(\rhi)
\end{equation}
where $S$ is the von-Neumann entropy. 
The proof of~\eqref{eq:CFD-bounds} is provided in Appendix~\ref{app:sec:upper-bounds::subsec:CFD-proof}.

In Eq.~\eqref{eq:CFD-bounds}, we observe that $\mathfrak{D}_c (\rhi) $ is upper bounded by $2C_{re}(\rhi)$, which is a valid measure of coherence of quantum state $\rhi$. Here, we note that $C_{re}(\rhi)$ depends on the initial coherent state $\rhi$, but is independent of $\Phi$ and $H(t)$ of the EPM protocol. As a result, the protocol-independent upper bound provided by $2C_{re}(\rhi) $ is rather weaker. 
On the other hand, the quantity $D_{KL} (\pcoh^{l,k}(\rhi) || \pcoh^{l,k}(\Delta [\rhi]))$ which depends on $\rhi$, $\Phi$, and $H(t)$ provides a tighter bound. Below we illustrate this with an example where we show that the value of the upper bound $D_{KL} (\pcoh^{l,k}(\rhi) || \pcoh^{l,k}(\Delta [\rhi]))$ can be very close to $\mathfrak{D}_c (\rhi)$ for single qubit states.

{As discussed above, the coherence fluctuation distance $\mathfrak{D}_c (\rhi) $ is process dependent: in general, it depends on the dynamics $\Phi$ and on the driving Hamiltonian $H(t)$. It is therefore useful to identify broad and physically relevant classes of open-system dynamics for which $\mathfrak{D}_c (\rhi) $ becomes trivial. For single-qubit systems, an important and widely used class is given by phase-covariant evolutions~\cite{Filippov:2019zdz, Holevo96Coh}, i.e., evolutions that commute with phase rotations generated by the (instantaneous) Hamiltonian. A key structural consequence of phase covariance is a decoupling between populations and coherences in the energy basis: the population dynamics is insensitive to initial coherences, and coherences do not feed into populations. In the EPM protocol, this directly constrains the possible dependence of the joint endpoint statistics $\pcoh^{l,k}$ on initial coherence and leads to a vanishing coherence fluctuation distance.}

{
\begin{definition}[\emph{Phase-covariant qubit dynamics}~\cite{Filippov:2019zdz}]
Without loss of generality, let us fix the reference (energy) basis as the eigenbasis of $\sigma_z$. A qubit dynamical map $\Phi$ is called {phase covariant} (with respect to $\sigma_z$) if it commutes with all phase rotations about the $z$-axis, i.e.,
\begin{equation}
e^{-i\phi \sigma_z}\Phi[\rho]e^{i\phi \sigma_z} =
\Phi\left[e^{-i\phi \sigma_z}\rho e^{i\phi \sigma_z}\right]
\qquad \forall \phi\in\mathbb{R}.
\end{equation}
Equivalently, in the Bloch-vector representation the evolution preserves the block structure between the transverse $(x,y)$ components and the longitudinal $z$ component: the $(x,y)$ components may undergo a rotation and a (generally contractive) linear transformation within the transverse plane, while the $z$ component evolves independently, with no mixing between the transverse and longitudinal sectors.
\end{definition}
}
{
In the EPM protocol, the relevant reference basis is the instantaneous energy eigenbasis of the system Hamiltonian. For a qubit, any Hamiltonian can be written as

\begin{equation}
    H(t) = \varepsilon_0 \mathbb{I}_2 + \frac{\omega(t)}{2} \textbf{n}(t). \boldsymbol{\sigma}
\end{equation}
with $\textbf{n}(t)$ a unit Bloch vector. Phase covariance with respect to the Hamiltonian means covariance under the one-parameter group generated by $H(t)$ , i.e., commutation with phase rotations about the instantaneous energy axis. Requiring this covariance to hold for all times $t$ singles out the case where the energy axis is fixed, $\textbf{n}(t) = \textbf{n}$, while the gap $\omega(t)$ may vary in time. Hence, up to an irrelevant energy shift, the most general qubit Hamiltonian compatible with phase-covariant dynamics can be taken as
\begin{equation}
    H(t) = \varepsilon_0 \mathbb{I}_2 + \frac{\omega(t)}{2} \sigma_z
\end{equation}
after choosing the energy eigenbasis so that $\textbf{n} = \hat{\textbf{z}}$. Therefore, working in the $\sigma_z$ eigenbasis is exactly the same as working in the instantaneous energy basis required by the EPM measurements.

\begin{theorem}~\label{them:1}
    The coherence fluctuation distance, $\mathfrak{D}_c (\rhi)$, identically vanishes 
    for a qubit system undergoing phase-covariant dynamics in the eigenbasis of the (time-dependent) Hamiltonian (taken diagonal in that basis).
\end{theorem}
\begin{proof}
We parameterize the initial coherent single-qubit state and its dephased (diagonal) counterpart in the $\sigma_z$ eigenbasis as
\begin{equation}\label{eq:rho-coh-single-qubit}
    \rhi = \begin{pmatrix}
a & \gamma \\
\gamma & 1-a 
\end{pmatrix}, \quad \gamma^2 \leq a (1-a), \; \; 0 \leq a \leq 1
\end{equation}
and
\begin{equation}\label{eq:rho-dephased-single-qubit}
    \Delta[\rhi] = \begin{pmatrix}
a & 0 \\
0 & 1-a 
\end{pmatrix}
\end{equation}
Since $\rhi$ and $ \Delta[\rhi]$ have identical diagonal entries in the energy basis, the initial virtual energy outcome distribution $\pinl$ is same for both the states.

Now, assume the dynamics $\Phi$ is phase covariant with respect to this energy basis. {Then, by the structural property of phase-covariant qubit channels, the evolved {$z$–component} (hence the diagonal part, i.e., the populations in the energy basis) depends only on the initial $z$–component and is independent of the initial {$x-y$ components} (coherences).} Equivalently, phase covariance implies that the diagonal elements of $\Phi[\rhi]$ in the $\sigma_z$ basis depend only on the diagonal elements of $\rhi$.

Therefore, at any time $t$, the diagonal parts of $\Phi[\rhi]$ and $\Phi[\Delta[\rhi]]$ coincide, and hence the final energy distribution $\pfink$ is the same for $\rhi$ and $\Delta[\rhi]$. Consequently, the EPM joint distribution $\pcoh^{l,k} = \pinl \pfink$ is identical for $\rhi$ and $\Delta[\rhi]$, for any choice of final time $\tf$, which implies
\begin{equation*}
    D_{KL} (\pcoh^{l,k}(\rhi) || \pcoh^{l,k}(\Delta [\rhi])) = 0
\end{equation*}
Using Eq.~\eqref{eq:CFD-bounds}, we conclude thats $\mathfrak{D}_c (\rhi) = 0$.
\end{proof}
Theorem~\eqref{them:1} implies that within the broad class of dynamics described by phase-covariant maps, initial coherence is thermodynamically irrelevant at the level of EPM trajectory statistics~\cite{PhysRevA.96.032109}.
}


\textit{Illustration 1.} Let us 
now compute all three quantities in Eq.~\eqref{eq:CFD-bounds} for a general single-qubit coherent state $\rhi$ and for particular choices of $H(t)$, and $\Phi$ and demonstrate that $D_{KL} (\pcoh^{l,k}(\rhi) || \pcoh^{l,k}(\Delta [\rhi])) $ indeed provides a strong upper bound for $\mathfrak{D}_c (\rhi)$.

Let the initial single qubit coherent state is 
\begin{equation}\label{eq:ini-coh}
    \rhi = \begin{pmatrix}
a & \gamma \\
\gamma & 1-a 
\end{pmatrix}
\end{equation}
where $\gamma^2 \leq a(1-a) $ is the coherence and $0\leq a \leq 1$.

First, let us consider only unitary dynamics for the system. In this case, the non-equilibrium statistics of  change in energy become non-equilibrium work statistics.

We consider the following time-dependent Hamiltonian
\begin{align}\label{eq:Ham-CFD-Unitary}
    H(t) = \frac{\Omega}{2} \left( \sin(\omega t) \sigma_x + \cos(\omega t) \sigma_z \right).
\end{align}
The three quantities in Eq.~\eqref{eq:CFD-bounds} are plotted in Fig.~\ref{fig:CFD-unitary} for this case for a representative set of parameter values. This numerical analysis validates the upper bounds presented in Eq.~\eqref{eq:CFD-bounds}.


{\textit{Illustration 2.}} Now, let us consider the open dynamics of the single qubit system. For this case the statistics of non-equilibrium energy fluctuations have contributions from both work and heat. We consider the dynamics described by Lindblad master equation
\begin{equation}\label{eq:master-eqn}
    \frac{d \rho}{d t} = - \mathrm{i} [H, \rho] + \kappa \Big( L \rho L^\dagger - \frac{1}{2} \{ L^\dagger L, \rho \} \Big)
\end{equation}
with a time-dependent Hamiltonian given by Eq.~\eqref{eq:Ham-CFD-Unitary} and Lindblad operator $L = \sigma_x$. The three quantities in Eq.~\eqref{eq:CFD-bounds} are plotted in Fig.~\ref{fig:CFD-dissipative} for this case for a representative set of parameter values. In both the Figs.~\ref{fig:CFD-unitary} and \ref{fig:CFD-dissipative}, we observe that CFD is a monotonically increasing function of coherence $|\gamma|$ and vanishes for incoherent initial states. 

\begin{figure}
    \centering
    \includegraphics[width=0.8\linewidth]{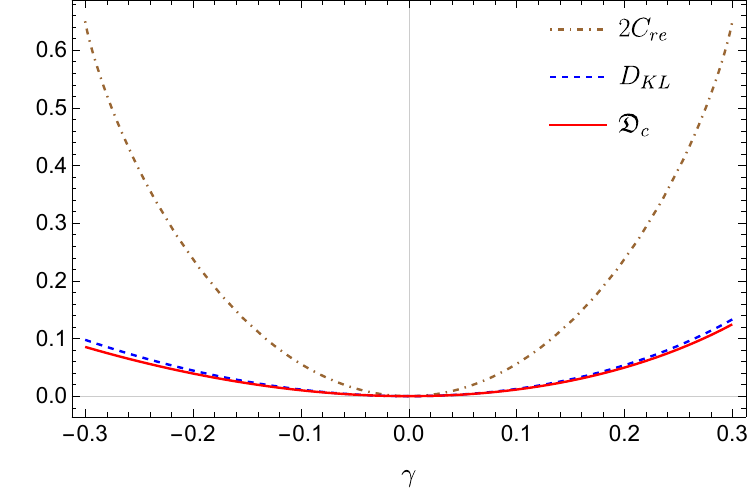}
    \caption{\textbf{\textit{Coherence fluctuation distance and its upper bounds for single-qubit initial states: unitary dynamics.} }
    The single qubit initial state is given by Eq.~\eqref{eq:ini-coh}, parametrized by coherence strength $\gamma$. The coherence fluctuation distance $\mathfrak{D}_c (\rhi)$ (solid red) is plotted as a function of $\gamma$ together with its two upper bounds from Eq.~\eqref{eq:CFD-bounds}: the intermediate bound (blue dashed), and the protocol-independent maximal bound $2 C_{re}(\rhi)$ (brown dot–dashed). The time evolution is unitary under the time-dependent Hamiltonian of Eq.~\eqref{eq:Ham-CFD-Unitary}. Parameters values used are  $\ti = 0$, $\tf = 10$, $a = 0.9$, $\Omega=1$, $\omega=1$.
    }
    \label{fig:CFD-unitary}
\end{figure}


\begin{figure}
  \centering
    \includegraphics[width=0.8\linewidth]{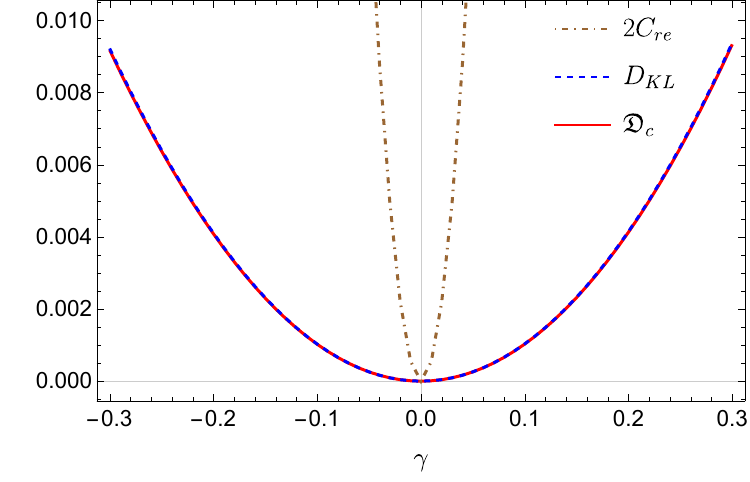}
  \caption{\textbf{\textit{Coherence fluctuation distance and its upper bounds for single-qubit initial states: dissipative dynamics.} }
  The single qubit initial state is given by Eq.~\eqref{eq:ini-coh}, parametrized by coherence strength $\gamma$. The coherence fluctuation distance $\mathfrak{D}_c (\rhi)$ (solid red) is plotted as a function of $\gamma$ together with its two upper bounds from Eq.~\eqref{eq:CFD-bounds}: the intermediate bound (blue dashed), and the protocol-independent maximal bound $2 C_{re}(\rhi)$ (brown dot–dashed). The time evolution is governed by Lindblad master equation, Eq.~\eqref{eq:master-eqn}, with the time-dependent Hamiltonian given by Eq.~\eqref{eq:Ham-CFD-Unitary} and Linblad operator $L = \sigma_x$. Parameters values used are  $\ti = 0$, $\tf = 10$, $a = 0.9$, $\Omega=1$, $\omega=1$, $\kappa = 0.1$.
  }
  \label{fig:CFD-dissipative}
\end{figure}



\subsection{Entanglement Fluctuation Distance}\label{sec:new-measures::subsec-EFD}

In Sec.~\ref{sec:new-measures::subsec-CFD} we introduced the CFD, which quantifies the minimal statistical distinguishability between the joint EPM energy trajectory distributions generated by a coherent initial state and those produced by any incoherent state. We now construct the entanglement analogue of this quantity, applicable to bipartite systems undergoing the EPM protocol.

Let $\mathscr{S}$ denote the set of separable bipartite states in the full Hilbert space $\mathcal{H}=\mathcal{H}_A \otimes \mathcal{H}_B$ of a $d$-dimensional bipartite quantum system $\sys$. For a fixed CPTP map $\Phi$ and bipartite Hamiltonian $H(t)$, the EPM distribution generated by an entangled initial state $\rhi$ will, in general, differ from those achievable using separable states. This motivates the following definition
\begin{equation}\label{eq:EFD-defn}
  \mathfrak{D}_E (\rhi) := \underset{\sigma \in \mathscr{S}}{\min} D_{KL} ( \pent^{\lb,\kb}(\rhi)|| \pent^{\lb,\kb}(\sigma))
\end{equation}
where $D_{KL}$ is the KL-divergence defined in Eq.~\eqref{eq:DKL-pEPM-coh-n-I}. The minimization over $\sigma \in \mathscr{S}$ ensures that $ \mathfrak{D}_E (\rhi) $ quantifies only the {entanglement–induced} contribution to the EPM statistics:
\begin{align}
    \mathfrak{D}_E (\rhi) =0  \Leftrightarrow \exists \sigma \in \mathscr{S} \text{ such that }  \pent^{\lb, \kb} (\rhi) = \pent^{\lb, \kb}(\sigma),
\end{align}
i.e., the entangled and separable processes become indistinguishable at the level of energy trajectories. Conversely, $\mathfrak{D}_E (\rhi) > 0$ implies that no separable initial state reproduces the same nonequilibrium energy statistics, thereby showing that entanglement has a thermodynamically active signature for this protocol.

Although $\mathfrak{D}_E(\rho_i)$ is defined at the level of total distribution, the bipartite fluctuation theorem shows that entanglement enters entropy production at the level of individual trajectories through the correction term $\D \Xilbkb$. The latter captures how entanglement affects entropy production in each single realization, while the former quantifies how much the overall EPM statistics deviate from those generated by separable states. Together, they describe complementary aspects of entanglement-driven nonequilibrium behavior.

The minimization in Eq.~\eqref{eq:EFD-defn} is difficult to evaluate in general. However, two useful families of upper bounds follow by comparison to separable reference states. The first is obtained from the separable component $\rhS$ of the BSA decomposition (Eq.~\eqref{eq:ent-ini-forward-BSA}), yielding
\begin{equation}\label{eq:EFD-upper-bound-1}
  \mathfrak{D}_E (\rhi) \leq D_{KL}(\pent^{\lb, \kb} (\rhi) || \pent^{\lb, \kb} (\rhS)) \leq 2 D(\rhi || \rhS)
\end{equation}
A second bound involves the relative entropy of entanglement of the entangled initial state $\rhi$, defined as
\begin{equation}
    D_E(\rhi):= \underset{\sigma \in \mathscr{S}}{\min} D(\rhi || \sigma)
\end{equation}
where $D(\cdot || \cdot)$ is the quantum relative entropy. Let $\rhstar$ denote the optimizer of $D_E(\rhi)$. Then
\begin{equation}\label{eq:EFD-upper-bound-2}
   \mathfrak{D}_E (\rhi) \leq D_{KL}(\pent^{\lb, \kb} (\rhi) || \pent^{\lb, \kb} (\rhstar)) \leq 2 D_E (\rhi)
\end{equation}
Together, Eqs.~\eqref{eq:EFD-defn}, \eqref{eq:EFD-upper-bound-1}, and \eqref{eq:EFD-upper-bound-2} establish an operational hierarchy linking the statistical impact of entanglement in energy–fluctuation experiments to standard entanglement monotones.
The proof of~\eqref{eq:EFD-upper-bound-1} and ~\eqref{eq:EFD-upper-bound-2} is provided in Appendix~\ref{app:sec:upper-bounds::subsec:CFD-proof}.

 To summarize this section, we have defined quantities based on the KL-divergence, that measure the contribution of coherence and entanglement {to non-equilibrium thermodynamic observables of the system.}

\section{Conclusions}\label{sec:conclusion}

In this work we developed a fully state-only, operational reformulation of fluctuation theorems within the  EPM framework.
This is in contrast to
previous EPM analyses, which relied on decompositions of the initial state into a thermal component and an auxiliary operator (coherence operator $\chi$)n which is not a physical state, not uniquely defined, and cannot be prepared experimentally. 
Our approach 
introduces 
experimentally preparable density matrices obtained through resource-theoretic decompositions of the initial state. Every element entering the corrected fluctuation relations is therefore a function of genuine quantum states, making all thermodynamic contributions explicitly accessible in an experiment.

For single systems, we introduced a decomposition of the initial state into a thermal part, a diagonal-athermal part, and a minimally coherent part, using the weight of athermality and weight of coherence. While for bipartite systems, we showed that nonlocal thermodynamic contributions can be disentangled using a fully state-restricted decomposition based on the Best Separable Approximation (BSA). This yields a refined structure for entropy production in which classical uncertainty, athermality,  coherence, classical correlations, and entanglement appear as separately identifiable and operationally meaningful contributions. The resulting Jarzynski equality and Crooks-type detailed fluctuation theorem cleanly resolve the influence of each quantum resource on the trajectory statistics of energy changes.

Building on these refined fluctuation relations, we introduced two new measures, 
the coherence fluctuation distance  and the entanglement fluctuation distance, defined as KL-divergences between EPM trajectory distributions generated by physical initial states. These quantify 
how strongly coherence or entanglement bias the statistics of energy changes relative to their optimal resource-free surrogates. 
%

Several directions naturally extend from the present work. First, the resource-resolved EPM framework can be combined with measurement-based feedback control, where intermediate measurements feed information into subsequent system evolution. 
Second, our bipartite analysis invites multipartite extensions, including EPM fluctuation relations on networks or chains of quantum systems with initial correlations distributed across different partitions. 
The state-only nature of all decompositions makes the present framework well suited for experimental implementation. 
The resulting EPM trajectory distributions can be measured on superconducting qubits, trapped ions, or other controllable platforms. These possibilities suggest that resource-resolved quantum fluctuation theorems may soon be testable in laboratory settings.

\acknowledgements
We acknowledge HPC facility of HRI. 
The research of S. Mondal was supported by the INFOSYS scholarship.

\appendix

\section{Weight of Athermality}\label{app:weight_atherm}


In this Appendix, we justify the explicit expressions in
Eqs.~\eqref{eq:weight_atherm_solns-a}and~\eqref{eq:weight_atherm_solns-b} for the weight of athermality and the
associated minimal athermal state. Throughout, we work
in finite dimensions and assume that the reference thermal
state $\gamma_{\beta,H}$ is full rank (which holds for any
finite, non-zero inverse temperature~$\beta$).

Recall Def.~\ref{def:athermality}: for a given state $\rho$ the weight of
athermality with respect to $\gamma_{\beta,H}$ is defined as
\begin{equation}
  A_w(\rho) := \min_{\tau\in\mathscr{D}(\mathcal H)}
  \{ a \ge 0 : \rho = (1-a)\,\gamma_{\beta,H} + a\,\tau \},
\end{equation}
where $\mathscr{D}(\mathcal H)$ denotes the set of density
operators on $\mathcal H$. In the main text, we stated that
the optimal weight $a$ and the corresponding minimal
athermal state $\tau$ are uniquely given by
\begin{subequations}
\begin{align}
  a &= 1 - \mu_{\min}\!\left(
      \gamma_{\beta,H}^{1/2}\,\rho\,\gamma_{\beta,H}^{-1/2}
      \right), \label{eq:Aw-appendix-a} \\
  \tau &= \frac{\rho - (1-a)\,\gamma_{\beta,H}}{a},
  \label{eq:Aw-appendix-tau}
\end{align}
\end{subequations}
where $\mu_{\min}(A)$ denotes the minimum eigenvalue
of a Hermitian operator $A$. We now prove existence,
optimality, and uniqueness of this decomposition.

\begin{lemma}
Let $\gamma_{\beta,H}$ be a full-rank thermal state. Then,
for every state $\rho$ there exists a unique $a\in[0,1]$ and
a unique $\tau\in\mathscr{D}(\mathcal H)$ such that
\begin{equation}
  \rho = (1-a)\,\gamma_{\beta,H} + a\,\tau,
\end{equation}
and the pair $(a,\tau)$ is given by
Eqs.~\eqref{eq:Aw-appendix-a}~and~\eqref{eq:Aw-appendix-tau}.
Moreover, $a$ is the minimal value compatible with such a
decomposition.
\end{lemma}

\begin{proof}
We first observe that any decomposition of the form
\begin{equation}
  \rho = (1-\tilde{\mathrm{a}})\,\gamma_{\beta,H} + \tilde{\mathrm{a}}\,\tau,
  \qquad \tau\in\mathscr{D}(\mathcal H),
  \label{eq:Aw-appendix-decomp}
\end{equation}
can be rewritten as
\begin{equation}
  \tau = \frac{\rho - (1-\tilde{\mathrm{a}})\,\gamma_{\beta,H}}{\tilde{\mathrm{a}}}.
  \label{eq:Aw-appendix-tau-from-a}
\end{equation}
Since both $\rho$ and $\gamma_{\beta,H}$ are density
operators with unit trace, the trace condition
$\tr\tau = 1$ is automatically satisfied for any $\tilde{\mathrm{a}}>0$,
because
\begin{equation}
  \tr\tau
  = \frac{\tr\rho - (1-\tilde{\mathrm{a}})\tr\gamma_{\beta,H}}{\tilde{\mathrm{a}}}
  = \frac{1 - (1-\tilde{\mathrm{a}})}{\tilde{\mathrm{a}}} = 1.
\end{equation}
Therefore, the only non-trivial requirement for
$\tau$ to be a density operator is 
being positive semi-definite,
\begin{equation}
  \rho - (1-\tilde{\mathrm{a}})\,\gamma_{\beta,H} \succeq 0.
  \label{eq:Aw-positivity}
\end{equation}

Let us define $s := 1-\tilde{\mathrm{a}}$. Then the
condition~\eqref{eq:Aw-positivity} becomes
\begin{equation}
  \rho - s\,\gamma_{\beta,H} \succeq 0.
  \label{eq:Aw-positivity-s}
\end{equation}
We are interested in the largest $s$ such that
Eq.~\eqref{eq:Aw-positivity-s} holds; the minimal $\tilde{\mathrm{a}}$ is then $a = 1-s_{\max}$. Formally,
\begin{equation}
  s_{\max} := \max
  \{ s\in\mathbb R : \rho - s\,\gamma_{\beta,H} \succeq 0 \},
  \qquad a := 1-s_{\max}.
\end{equation}

Because $\gamma_{\beta,H}$ is full rank, we can conjugate
Eq.~\eqref{eq:Aw-positivity-s} by
$\gamma_{\beta,H}^{-1/2}$ and define
\begin{equation}
  X := \gamma_{\beta,H}^{-1/2}\,\rho\,\gamma_{\beta,H}^{-1/2}.
\end{equation}
Then Eq.~\eqref{eq:Aw-positivity-s} is equivalent to
\begin{equation}
  \gamma_{\beta,H}^{-1/2}
  \big(\rho - s\,\gamma_{\beta,H}\big)
  \gamma_{\beta,H}^{-1/2} \succeq 0
  \;\;\Longleftrightarrow\;\;
  X - s\,\mathbb I \succeq 0.
\end{equation}
{The latter condition holds if and only if the minimum
eigenvalue satisfies $\mu_{\min}(X) - s \ge 0$}, i.e.
\begin{equation}
  s \le \mu_{\min}(X).
\end{equation}
Hence, the maximal admissible value of $s$ is
\begin{equation}
  s_{\max} = \mu_{\min}(X)
  = \mu_{\min}\!\left(
    \gamma_{\beta,H}^{-1/2}\,\rho\,\gamma_{\beta,H}^{-1/2}
    \right),
\end{equation}
and the corresponding minimal $\tilde{\mathrm{a}}$ is
\begin{equation}
  a = 1 - s_{\max}
  = 1 - \mu_{\min}\!\left(
    \gamma_{\beta,H}^{-1/2}\,\rho\,\gamma_{\beta,H}^{-1/2}
    \right),
\end{equation}
which is Eq.~\eqref{eq:Aw-appendix-a}. This shows both
existence (we can construct such an $\tilde{\mathrm{a}}$) and optimality
(any $\tilde{\mathrm{a}}'<a$ would correspond to an $s'>s_{\max}$
and violate positivity).

We next show that $a\in[0,1]$. Since $X$ is
positive semi-definite, all its eigenvalues are non-negative, and
$\mu_{\min}(X)\ge 0$, implying $a\le 1$.
Moreover, if $\rho=\gamma_{\beta,H}$ then
$X=\mathbb I$ and $\mu_{\min}(X)=1$, so $a=0$.
{If $\rho\neq\gamma_{\beta,H}$, then $\rho \nsucceq \gamma_{\beta,H}$. Consequently, $X \nsucceq \mathbb{I}$, implying that 
$\mu_{\min}(X)<1$ and therefore $a>0$.}
In any case, $0\le a\le 1$.

Finally, define $\tau$ by
Eq.~\eqref{eq:Aw-appendix-tau-from-a} with
$\tilde{\mathrm{a}}=a$. We have already seen that
$\tr\tau = 1$, and by construction
$\rho-(1-a)\gamma_{\beta,H} \succeq 0$, so
$\tau \succeq 0$. Hence $\tau\in\mathscr{D}(\mathcal H)$, and
\begin{equation}
  \rho = (1-a)\,\gamma_{\beta,H} + a\tau
\end{equation}
provides the desired decomposition. This proves the
existence of a minimal athermal state.

To prove uniqueness, first note that $s_{\max}$, and hence
$a$, is uniquely determined by the spectrum of $X$. 
{We would like to clarify here that the spectrum of $X$ can have degenerate minima, but the minimum value itself is unique.}
Suppose, for the same $a$, that we had two states
$\tau_1,\tau_2$ such that
\begin{equation}
  \rho = (1-a)\gamma_{\beta,H} + a\tau_1
      = (1-a)\gamma_{\beta,H} + a\tau_2.
\end{equation}
Then $a(\tau_1 - \tau_2) = 0$, which implies
$\tau_1=\tau_2$ whenever $a>0$. When $a=0$,
we have $\rho=\gamma_{\beta,H}$ and the decomposition is
trivial; in that case, we may conventionally set
$\tau=\gamma_{\beta,H}$. Thus, for any athermal state
$\rho\neq\gamma_{\beta,H}$, both the minimal weight
$a$ and the minimal athermal state $\tau$ are
unique.
\end{proof}
The above argument justifies the expressions of Eqs.~\eqref{eq:weight_atherm_solns-a}~and~\eqref{eq:weight_atherm_solns-b} quoted in the main text and shows that, under the
assumption of a full-rank thermal reference state, the
weight of athermality $A_w(\rho)$ is a well-defined and
operationally meaningful quantity.

\section{Derivation of Jarzynski Equalities}

The EPM Jarzynski equality can be readily derived from the EPM characteristic function. The expression for the EPM characteristic function, $\mathcal{G}(u)$, is given in Eq.~\eqref{eq:EPM-characteristic} of the main text. Recall,
\begin{align}
    \mathcal{G}(u) := \langle e^{\mathrm{i} u \D E} \rangle = \tr \left(  \rhi e^{- \mathrm{i} u \Hi}\right) \tr \left( \Phi[\rhi] e^{\mathrm{i} u \Hf} \right)
\end{align}
Putting $u = \mathrm{i} \B$, we get
\begin{align}\label{eq:Jarzynski-base}
     \langle e^{-\B \D E} \rangle &=  \tr \left(  \rhi e^{\B \Hi}\right) \tr \left( \Phi[\rhi] e^{- \B \Hf} \right) \nonumber \\
  \Rightarrow   \langle e^{-\B \D E} \rangle &= \frac{\zbf}{\zbi}  \tr \left(  \rhi \left( \gbi \right)^{-1}\right) \tr \left( \Phi[\rhi] \gbf \right)  \nonumber \\
 \Rightarrow  \langle e^{-\B \left( \D E - \D F \right)} \rangle &= \tr \left(  \rhi \left( \gbi \right)^{-1}\right) \tr \left( \Phi[\rhi] \gbf \right)
\end{align}
where we have used $\D F = - \B^{-1} \ln \left(\zbf/\zbi \right)$. By substituting various resource-theoretic decompositions of the initial state $\rhi$ in Eq.~\eqref{eq:Jarzynski-base}, one can obtain the respective Jarzynski equalities.

\subsection{Athermality}\label{app:Jarzynski:subsec:atherm}
By substituting $\rhi = (1 - a)\gbi + a \tau$ in Eq.~\eqref{eq:Jarzynski-base}, we obtain the athermality-corrected Jarzynski equality,
\begin{align}\label{eq:Jarzynski-atherm-app}
    &\Big\langle e^{-\B (\D E - \D F)}\Big \rangle = \nonumber \\ & \left\{ (1-a) d + a  \tr ( (\gbi)^{-1} \tau ) \right\}  \nonumber \\ 
      &\left \{ (1-a) \tr ( \gbf \Phi[\gbi] ) 
    + a \tr (\gbf \Phi [\tau] ) \right \}
\end{align}

\subsection{Coherence}\label{app:Jarzynski:subsec:coh}
By substituting $\tau = (1 -c ) \taud + c \tauc$ in Eq.~\eqref{eq:Jarzynski-atherm-app}, we obtain the coherence-corrected Jarzynski equality,
\begin{align}\label{eq:Jarzynski-coh-app}
    &\langle e^{-\B \left( \D E  - \D F\right)} \rangle = \nonumber \\
    &\Big \{ (1 - a)d + a (1 - c)  \tr((\gbi)^{-1} \taud)  + a c  \tr ((\gbi)^{-1}  \tauc)\Big \} \nonumber \\
    &\Big \{  (1 - a) \tr (\gbf \Phi[\gbi])  + a (1 - c) \tr (\gbf \Phi[\taud]) \nonumber \\
    &+ a c \tr (\gbf \Phi[\tauc]) \Big \}
\end{align}

\subsection{Entanglement via $\mathfrak{E}_{AB}$ decomposition}\label{app:Jarzynski:subsec:biparty-EnAB}
Using the decomposition $\rhi = \rhAi \otimes \rhBi + \EnAB$ of the bipartite initial state $\rhi$ in Eq.~\eqref{eq:Jarzynski-base}, the corresponding correlation-corrected Jarzynski equality can be obtained as
\begin{align}\label{eq:Jarzynski-EnAB-biparty-app}
    &\langle e^{-\B (\D E  - \D F)} \rangle = \nonumber \\
    &\left \{d +  \tr ( (\gbi)^{-1} \EnAB) \right \} \nonumber \\
    &\Big \{ \tr(\gbf \Phi[\gbi]) 
     + \tr(\gbf \Phi[\EnAB] ) \Big \}
\end{align}

\subsection{Entanglement via Best Separable Approximation}\label{app:Jarzynski:subsec:biparty-BSA}
Any bipartite initial state $\rhi$ can be decomposed as
\begin{equation}\label{eq:BSA-app}
    \rhi =  \lambda \rhE + (1-\lambda) \rhS,
\end{equation}
via the Best Separable Approximation (BSA) introduced in Sec.~\eqref{sec:EPM-FT-biparty::subsec:decomp-II-BSA} of the main text (see Def.~\eqref{def:BSA}). Furthermore, the separable state $\rhS$ admits the following decomposition in terms of the product states,
\begin{equation}\label{eq:BSA-seprable-app}
    \rhS = \sum_j \rj \rhAj \otimes \rhBj, \quad 0\leq \rj \leq 1,
\end{equation}
Each local state in $\rhS$ can be further decomposed using the weight of athermality and the weight of coherence introduced in Sec.~\ref{sec:EPM-FT},
\begin{align}\label{eq:ent-ini-forward-BSA-local-state-decomp}
    \rhAj &= (1 - \aAj) \gbAi + \aAj (1 - \cAj) \tauAdj + \aAj \cAj \tauAcj, \nonumber 
    \\
     \rhBj &= (1 - \aBj) \gbBi + \aBj (1 - \cBj) \tauBdj + \aBj \cBj \tauBcj 
\end{align}
where $a_j^{A(B)}$, $c_j^{A(B)}$ are, respectively, the weights of
athermality and coherence of $\rho_j^{A(B)}$ with respect to $\gamma_{\beta, i}^{A(B)}$, and $\tau_{d,j}^{A(B)}$ and $\tau_{c,j}^{A(B)}$ are
the corresponding minimal diagonal-athermal and minimal coherent states. For each $j$ in the sum, $\rhAj \otimes \rhBj$ is a sum of nine terms,
\begin{align}\label{eq:BSA-product-decomposition-explicit-expr}
    \rhAj \otimes \rhBj &=  (1 - \aAj) (1 - \aBj) \gbi + \rhdj + \rhcj,  \nonumber\\
     \gbi &=  \gbAi \otimes  \gbBi, \nonumber \\
     \rhdj &= (1 - \aAj) \aBj (1 - \cBj)  \gbAi \otimes \tauBdj  \nonumber  \\
     &+ \aAj (1 - \cAj) (1 - \aBj) \tauAdj \otimes \gbBi \nonumber \\
     &+ \aAj (1 - \cAj) \aBj (1 - \cBj) \tauAdj \otimes \tauBdj,  \nonumber \\
     \rhcj &=  (1 - \aAj) \aBj \cBj   \gbAi \otimes  \tauBcj \nonumber \\
     &+ \aAj \cAj (1 - \aBj) \tauAcj \otimes \gbBi \nonumber \\
     &+ \aAj (1 - \cAj) \aBj \cBj \tauAdj \otimes \tauBcj \nonumber \\
     &+  \aAj \cAj  \aBj (1 - \cBj) \tauAcj \otimes \tauBdj \nonumber \\
     &+ \aAj \cAj  \aBj \cBj \tauAcj \otimes \tauBcj
\end{align}
Using the BSA decomposition~\eqref{eq:BSA-app} together with the local
decompositions~\eqref{eq:BSA-product-decomposition-explicit-expr} in the
Eq.~\eqref{eq:Jarzynski-base}, the entanglement-corrected Jarzynski equality becomes
\begin{align}
    \langle e^{-\beta (\Delta E - \Delta F)} \rangle = \mathcal{J}^{(i)} \mathcal{J}^{(f)},
\end{align}
where
\begin{align} 
     & \mathcal{J}^{(i)} = \nonumber \\
     &  \lambda  \tr \left( \left( \gbi \right)^{-1} \rhE \right) +  (1 - \lambda) \Big \{ \sum_j \rj \Big( (1 - \aAj) (1 - \aBj) d  \nonumber\\
      &+  \tr \left( \left( \gbi \right)^{-1} \rhdj \right)  + \tr  \left( \left( \gbi \right)^{-1} \rhcj \right)  \Big) \Big \}   \nonumber\\
     & \mathcal{J}^{(f)} = \nonumber \\
     & \Big \{ \lambda \tr \left( \gbf \Phi [\rhE] \right) + \nonumber \\
     &(1 - \lambda) \Big \{  \sum_j \rj \Big( (1 - \aAj) (1 - \aBj) \tr \left( \gbf \Phi [\gbi] \right)\nonumber\\ 
     &  + \tr \left( \gbf  \Phi [\rhdj] \right)  + \tr \left( \gbf  \Phi [\rhcj] \right)  \Big)  \Big \} \Big \} \nonumber
\end{align}

\section{Derivation of Entropy-Production Fluctuation Theorems}
Recall that the trajectory-level entropy production is defined as 
\begin{align}
    \D \stotlk := \ln \left( \frac{P_\G(l,k)}{P_{\Gt}(k,l)} \right)
\end{align}
where $P_\G = \pinl \pfink$, and $P_{\Gt} = \ptink \ptfinl$ are the EPM probability distributions associated with the forward and backward processes, respectively.
\begin{align*}
    \pinl := \tr \left( \rhi \Piil \right), \qquad  \pfink := \tr \left( \Phi[\rhi] \Pifk \right), \\
    \ptink := \tr \left( \rhti \Pifk \right), \qquad  \ptfinl := \tr \left( \tilde{\Phi}[\rhti] \Piil \right)
\end{align*}
The Kraus operators of the dual time-reversed map $\tilde{\Phi}$ can be expressed in terms of those of the CPTP map $\Phi$ of the forward dynamics as described in the Sec.~\ref{sec:EPM-rev::subsec:FT} of the main text (see Eq.~\eqref{eq:Phi-dual-Kraus}). It is assumed that $\Phi$ admits a non-singular fixed point. $\rhi$ and $\rhti$ are initial states of the forward and reversed dynamics, respectively. For any generic linear operator $\mathcal{A}$, we can define,
\begin{align*}
    \pinl(\mathcal{A}) &\coloneq \tr(\mathcal{A} \; \Pi^i_l ), \quad 
    \pfink(\mathcal{A}) \coloneq \tr(\Phi(\mathcal{A}) \; \Pifk ), \\
   \ptink(\mathcal{A}) &\coloneq \tr(\mathcal{A} \; \Pifk ), \quad
   \ptfinl (\mathcal{A}) \coloneq \tr \left( \Phtil (\mathcal{A}) \; \Piil   \right),
\end{align*}

\subsection{Coherence}\label{app:Crooks:subsec:coh}
We decompose $\rhi$ and $\rhti$ using the weight of athermality and the weight of coherence as
\begin{subequations}
    \begin{align}
        \rhi &= (1 - a)\gbi + a (1-c) \taud + a c \tau_c  \label{state-convexMix-fwd-app}\\
        \rhti &= (1 - \atl)\gbf + \atl (1-\ctl) \tautd + \atl \ctl \tautc \label{state-convexMix-bkd-app}
    \end{align}
    \end{subequations}
    Plugging this decomposition of $\rhi$ in the expression for $\pinl$ gives us
    \begin{align}\label{eq:pinl-coh}
        \pinl &=  (1-a) \pinl(\gbi) + a (1-c) \pinl (\taud) + a c \pinl(\tauc) \nonumber \\
        & = \pinl(\gbi)e^{\Thil (\taud)}e^{\SBil (\tauc)}
    \end{align}
    where
   \begin{align*}
        \Thil(\taud) &:= \ln \left( (1-a) + a (1-c) \frac{ \pinl(\taud)  }{\pinl(\gbi)} \right) \\
        \SBil(\tauc) &:= \ln \Big( 1 +   
        \frac{ac \; \pinl (\tauc) }{(1-a) \pinl (\gbi) + a (1 - c) \pinl (\taud) } \Big)
   \end{align*}
   Similarly, we can get
   \begin{align}\label{eq:pfink-coh}
       \pfink &= \pfink(\gbi)e^{\Thfk (\taud)}e^{\SBfk (\tauc)}
   \end{align}
   where
    \begin{align*}
         \Thfk(\taud)&:=  \ln  \Big( (1-a) + a (1-c) \frac{\pfink (\taud)}{\pfink (\gbi)}  \Big) \\
        \SBfk(\tauc) &:=  \ln  \Big( 1 + 
 \frac{a c \; \pfink (\tauc)}{(1 -a)\pfink (\gbi) + a (1 - c) \pfink (\taud)} \Big)
   \end{align*}
   Combining Eq.~\eqref{eq:pinl-coh} and Eq.~\eqref{eq:pfink-coh}, we can express the probability of the forward process as
   \begin{align}
       P_\G (l,k) &= \pinl \pfink \nonumber \\
        &= p_l^i(\gamma_{\B,i}) p_k^f(\gamma_{\B,i}) e^{\Theta_{lk}^\G(\taud)}e^{ \Sigma_{lk}^\G(\tauc)}.
   \end{align}
    Similarly, the probability of the backward process can be expressed as 
   \begin{align}
     \PGt (k,l) &= \tilde p_k^i(\gamma_{\B,f}) \tilde p_l^f(\gamma_{\B,f}) e^{\Theta_{lk}^{\Gt}(\tautd)}e^{ \Sigma_{lk}^{\Gt}(\tautc)},
\end{align}
where the new quantities $\mathcal{M}_{lk}^\G$ and $\mathcal{M}_{lk}^{\Gt}$, with $\mathcal{M} = \{\Theta, \Sigma\}$, are defined by collecting the contributions arising from both the initial virtual measurement and the final projective measurement, as follows,
\begin{align*}
    \mathcal{M}_{lk}^\G = \mathcal{M}_l^{(i)} + \mathcal{M}_k^{(f)}\\
     \mathcal{M}_{lk}^{\Gt} = \mathcal{M}_l^{(i)} + \mathcal{M}_k^{(f)}.
\end{align*}
The explicit expressions for the quantities $\Thtik$, $\Thtfl$, $ \SBtik$, and $\SBtfl$ appearing in the expression for $\PGt$ above are as follows,
\begin{align*}
  \Thtik(\tautd) &:= \ln \left( (1-\atl) + \atl (1-\ctl) \frac{ \ptink(\tautd)  }{ \ptink(\gbf) } \right)\\
 \Thtfl(\tautd)&:=  \ln  \Big( (1 - \atl )  + \atl  (1 - \ctl) \frac{\ptfinl (\tautd)}{\ptfinl(\gbf)} \Big)  
 \end{align*}
\begin{align*}
   \SBtik(\tautc)&:= \ln \Big( 1 +    
   \frac{\atl \ctl \; \ptink(\tautc) }{(1-\atl) \ptink(\gbf)  + \atl (1 - \ctl) \ptink(\tautd) } \Big) \\
 \SBtfl(\tautc) &:=  \ln   \Big( 1 + 
 \frac{\atl \ctl \; \ptfinl (\tautc)}{(1 -\atl)\ptfinl (\gbf) + \atl (1 - \ctl) \ptfinl (\tautd)} \Big)
\end{align*}
The entropy production becomes
\begin{align}
    \D \stotlk &= \ln \left( \frac{P_\G(l,k)}{P_{\Gt}(k,l)} \right) \nonumber \\
    &= \ln \left( \frac{\pinl \pfink}{\ptink \ptfinl} \right) \nonumber \\
    &= \ln \left( \frac{p_l^i(\gamma_{\B,i}) p_k^f(\gamma_{\B,i})}{\tilde p_k^i(\gamma_{\B,f}) \tilde p_l^f(\gamma_{\B,f})}\right) + \D \Thlk + \SBlk \nonumber \\
    &= \B (\D \Elk - \D F) + \underbrace{\D \slk +  \D \Thlk + \SBlk}_{\D \scorlk}
\end{align}
where we have used
\begin{align*}
    \ln \left(\frac{p_l^i(\gamma_{\B,i})}{\tilde p_k^i(\gamma_{\B,f})} \right) = \B (\D \Elk - \D F), \quad \D \Elk = \Efk - \Eil
\end{align*}
and defined
\begin{align*}
     \sfk (\gbi) &:= \ln  \pfink (\gbi), \quad  \stfl(\gbf) := \ln \tilde{p}^f_l(\gbf), \\
     \D \slk &:= \s^f_{k} (\gbi) -  \st^f_l(\gbf),  \\
     \D \Thlk &:= \Theta_{lk}^\G(\taud) - \Theta_{lk}^{\Gt}(\tautd), \\
     \D \SBlk &:= \Sigma_{lk}^\G(\tauc) - \Sigma_{lk}^{\Gt}(\tautc)
\end{align*}

\subsection{Entanglement via $\mathfrak{E}_{AB}$ decomposition}\label{app:Crooks:subsec:biparty-EnAB}
We now consider the initial state decompositions
\begin{align*}
    \rhi &= \gbAi \otimes \gbBi + \EnAB \\
    \rhti &= \gbAf \otimes \gbBf + \EntAB 
\end{align*}
Before we derive the FT, we define the following quantities for the bipartite EMP protocol
\begin{align*}
    \pinlb(\mathcal{A}) &\coloneq \tr(\mathcal{A} \; \PiilA \otimes \PiilB ), \quad, \\
    \pfinkb(\mathcal{A}) &\coloneq \tr(\Phi(\mathcal{A}) \; \PifkA \otimes \PifkB ), \\
   \ptinkb(\mathcal{A}) &\coloneq \tr(\mathcal{A} \; \PifkA \otimes \PifkB ), \quad, \\
   \ptfinlb (\mathcal{A}) &\coloneq \tr \left( \Phtil (\mathcal{A}) \; \PiilA \otimes \PiilB   \right),
\end{align*}
for any generic linear operator $\mathcal{A}$. With the above initial state decomposition, we get,
\begin{align}\label{eq:Psiinlb-app}
    \pinlb &= \pinlb(\gbi) + \pinlb(\EnAB), \nonumber \\
    &= \pinlb(\gbi) e^{\Psi^{(i)}_{\lb} (\EnAB)}
\end{align}
where $\gbi = \gbAi \otimes \gbBi$ and
\begin{align*}
    \Psi^{(i)}_{\lb} (\EnAB) = \ln \left( 1 + \frac{ \pinlb(\EnAB)}{\pinlb(\gbi)} \right)
\end{align*}
Similarly, we can get
\begin{align}\label{eq:Psifinkb-app}
     \pfinkb  = \pfinkb (\gbi) e^{\Psi^{(f)}_{\kb} (\EnAB)}
\end{align}
where
\begin{align*}
    \Psi^{(f)}_{\kb} (\EnAB) = \ln \left( 1 + \frac{ \pfinkb(\EnAB)}{\pfinkb(\gbi)} \right)
\end{align*}
From Eq.~\eqref{eq:Psiinlb-app} and~\eqref{eq:Psifinkb-app}, we can express $P_\G$ as
\begin{equation}
    P_\G(\lb,\kb) = \pinlb(\gbi) \pfinkb (\gbi) e^{\Psi^{\G}_{\lb,\kb} }
\end{equation}
where 
\begin{align*}
    \Psi^{\G}_{\lb,\kb} = \Psi^{(i)}_{\lb} + \Psi^{(f)}_{\kb}
\end{align*}
Similarly, we can derive the expression for $P_{\Gt}$ as
\begin{equation}
    P_{\Gt}(\lb,\kb) = \ptinkb(\gbf) \ptfinlb (\gbf) e^{\Psi^{\Gt}_{\lb,\kb} }
\end{equation}
with
\begin{align*}
    \Psi^{\Gt}_{\lb,\kb} = \tilde{\Psi}^{(i)}_{\kb}(\EntAB) + \tilde{\Psi}^{(f)}_{\lb}(\EntAB)
\end{align*}
The expressions for $\tilde{\Psi}^{(i)}_{\kb}$ and $\tilde{\Psi}^{(f)}_{\lb}$ are
\begin{align*}
    \tilde{\Psi}^{(i)}_{\lb} (\EntAB) &= \ln \left( 1 + \frac{ \ptinkb(\EntAB)}{\ptinkb(\gbf)} \right) \\
    \tilde{\Psi}^{(f)}_{\kb} (\EntAB) &= \ln \left( 1 + \frac{ \ptfinlb(\EntAB)}{\ptfinlb(\gbf)} \right)
\end{align*}
Finally, the trajectory-level entropy production is
\begin{align}
    \D \stotlbkb &= \ln \left( \frac{P_\G(\lb,\kb)}{P_{\Gt}(\kb,\lb)} \right) \nonumber \\
    &= \ln \left( \frac{\pinlb \pfinkb}{\ptinkb \ptfinlb} \right) \nonumber \\
    &= \ln \left( \frac{\pinlb(\gamma_{\B,i}) \pfinkb(\gamma_{\B,i})}{ \ptinkb(\gamma_{\B,f})  \ptfinlb(\gamma_{\B,f})}\right) + \D \Psi_{\lb,\kb} \nonumber \\
    &= \B (\D \Elbkb - \D F) + \underbrace{ \D \sigma_{\lb, \kb} +\D \Psi_{\lb,\kb} }_{\D \scorlbkb} 
\end{align}
where we have used
\begin{align*}
    \ln \left(\frac{\pinlb(\gamma_{\B,i})}{ \ptinkb(\gamma_{\B,f})} \right) &= \B (\D \Elbkb - \D F), \\ 
    \D \Elbkb &= \EfkA + \EfkB - \EilA - \EilB
\end{align*}
and defined
\begin{align*}
 \sfkb (\gbi) &:= \ln  \pfinkb (\gbi), \quad  \stflb(\gbf) := \ln \tilde{p}^f_{\lb}(\gbf), \\
     \D \slbkb &:= \s^f_{\kb} (\gbi) -  \st^f_{\lb}(\gbf),  \\
     \D \Psi_{\lb,\kb} &:= \Psi_{\lb, \kb}^\G(\EnAB) - \Psi_{\lb, \kb}^{\Gt}(\EntAB).
\end{align*}

\subsection{Entanglement via Best Separable Approximation}\label{app:Crooks:subsec:biparty-BSA}
Finally, let us consider initial state decompositions given by BSA in Eq.~\eqref{eq:BSA-app},~\eqref{eq:BSA-seprable-app},~\eqref{eq:ent-ini-forward-BSA-local-state-decomp},~\eqref{eq:BSA-product-decomposition-explicit-expr} for $\rhi$, and assume a similar BSA decomposition for $\rhti$.

We get
\begin{align}\label{eq:pinlb-BSA}
    & \pinlb \nonumber = \lambda \pinlb(\rhE) + \nonumber\\ 
    &(1 - \lambda) \sum_j \rj \left \{ (1 - \aAj) (1 - \aBj) \pinlb(\gbi) + \pinlb (\rhdj) + \pinlb (\rhcj) \right \} \nonumber \\
    &= \pinlb(\gbi) e^{\Lamilb} e^{\Xiilb}
\end{align}
where
\begin{align}
    \Lamilb &\coloneqq \ln \left( \sum_j \rj \SSilbj \right), \\
    \SSilbj &\coloneqq (1-\aAj)(1 - \aBj) + \frac{\pinlb(\rhdj)}{\pinlb(\gbi)} + \frac{\pinlb(\rhcj)}{\pinlb(\gbi)}, \\
     \Xiilb &\coloneqq (1 - \lambda) + \lambda \frac{\pinlb(\rhE)}{\pinlb(\gbi) \Lamilb}
\end{align}
and
\begin{align}\label{eq:pfinkb-BSA}
    \pfinkb &= \pfinkb(\gbi) e^{\Lamfkb} e^{\Xifkb}
\end{align}
where
\begin{align}
    \Lamfkb &\coloneqq  \ln \left(\sum_j \rj \SSfkbj \right), \\
     \SSfkbj &\coloneqq (1-\aAj)(1 - \aBj) + \frac{\pfinkb(\rhdj)}{\pfinkb(\gbi)} + \frac{\pfinkb(\rhcj)}{\pfinkb(\gbi)}, \\
     \Xifkb &\coloneqq (1 - \lambda) + \lambda \frac{\pfinkb(\rhE)}{\pfinkb(\gbi) \Lamfkb}
\end{align}
Combining Eq.~\eqref{eq:pinlb-BSA} and~\eqref{eq:pfinkb-BSA}, we can get the expression for $P_\G$ as
\begin{align}
    P_\G(\lb, \kb) = \pinlb(\gbi) \pfinkb(\gbi) e^{\Lambda^{\G}_{\lb, \kb}} e^{\Xi^{\G}_{\lb, \kb}}
\end{align}
where 
\begin{align*}
    \Lambda^{\G}_{\lb, \kb} &= \Lamilb + \Lamfkb,  \\
    \Xi^{\G}_{\lb, \kb}  &= \Xiilb +  \Xifkb 
\end{align*}
Analogously, we can get
\begin{align}
     P_{\Gt}(\lb, \kb) = \ptinkb(\gbf) \ptfinlb(\gbf) e^{\Lambda^{\Gt}_{\lb, \kb}} e^{\Xi^{\Gt}_{\lb, \kb}}
\end{align}
with
\begin{align*}
    \Lambda^{\Gt}_{\lb, \kb} &= \Lamtikb + \Lamtflb,  \\
    \Xi^{\Gt}_{\lb, \kb}  &= \Xitikb +  \Xitflb 
\end{align*}
The explicit expressions of $\Lamtikb$, $\Lamtflb$, $\Xitikb$, $\Xitflb$ are as follows
\begin{align*}
    \Lamtikb &\coloneqq  \ln \left(\sum_{\jt} \rtj \SStikbj \right), \\
     \SStikbj &\coloneqq (1-\atAj)(1 - \atBj) + \frac{\ptinkb(\rhtdj)}{\ptinkb(\gbf)} + \frac{\ptinkb(\rhtcj)}{\ptinkb(\gbf)}, \\
     \Xitikb &:= (1 - \lamt) + \lamt \frac{\ptinkb(\rhtE)}{\ptinkb(\gbf) \Lamtikb},
\end{align*}
\begin{align*}
     \Lamtflb &\coloneqq \ln \left( \sum_{\jt} \rtj \SStflbj \right), \\
     \SStflbj &\coloneq (1-\atAj)(1 - \atBj) + \frac{\ptfinlb(\rhtdj)}{\ptfinlb(\gbf)} + \frac{\ptfinlb(\rhtcj)}{\ptfinlb(\gbf)}, \\
     \Xitflb &\coloneq (1 - \lamt) + \lamt \frac{\ptfinlb(\rhtE)}{\ptfinlb(\gbf) \Lamtflb}
\end{align*}
The entropy production can be obtained as
\begin{align}
     \D \stotlbkb &= \ln \left( \frac{P_\G(\lb,\kb)}{P_{\Gt}(\kb,\lb)} \right) \nonumber \\
    &= \ln \left( \frac{\pinlb \pfinkb}{\ptinkb \ptfinlb} \right) \nonumber \\
    &= \ln \left( \frac{\pinlb(\gamma_{\B,i}) \pfinkb(\gamma_{\B,i})}{ \ptinkb(\gamma_{\B,f})  \ptfinlb(\gamma_{\B,f})}\right)+ \D \Lambda_{\lb,\kb} + \D \Xi_{\lb,\kb} \nonumber \\
    &= \B (\D \Elbkb - \D F) + \underbrace{+ \D \Lambda_{\lb,\kb} + \D \Xi_{\lb,\kb}}_{\D \scorlbkb} 
\end{align}

\section{Upper bound proofs}
\subsection{Upper bounds of CFD}\label{app:sec:upper-bounds::subsec:CFD-proof}

We now prove the inequality relation of Eq.~\eqref{eq:CFD-bounds},
\begin{equation*}
    \mathfrak{D}_c (\rhi) \leq D_{KL} (\pcoh^{l,k}(\rhi) || \pcoh^{l,k}(\Delta [\rhi])) \leq 2 C_{re}(\rhi)
\end{equation*}

\begin{proof}
To begin with, note that $C_{re}(\rhi)$, defined in Eq.~\eqref{eq:Cre-defn}, can be alternatively expressed as 
\begin{equation}\label{eq:Cre_as_relative_ent}
    C_{re}(\rhi) = D(\rhi || \Delta(\rhi))
\end{equation}
where $D(\rho || \sigma) := \tr \left( \rho (\ln \rho - \ln \sigma) \right)$ is the relative entropy between two quantum states $\rho$ and $\sigma$. $\Delta(\rhi)$ is the diagonal part of state $\rhi$ in the eigenbasis of the initial Hamiltonian $\Hi$.

Recall that the EPM probability distribution $\pcoh^{l,k}$ is a product of two probability distributions $\pinl$ and $\pfink$ associated with initial virtual and final projective energy measurements, respectively. Then from the definition of $D_{KL} (\pcoh^{l,k}(\rhi)  || \pcoh^{l,k}(\rhi^{\mathcal{I}}))$ in Eq.~\eqref{eq:DKL-pEPM-coh-n-I}, it immediately follows that 
\begin{align}\label{eq:coh_DKL_in_plus_fin}
    &D_{KL} (\pcoh^{l,k}(\rhi)  || \pcoh^{l,k}(\Delta(\rhi))) = \nonumber \\
    & D_{KL} (\pinl(\rhi)  || \pinl(\Delta(\rhi))) +  D_{KL} (\pfink(\rhi)  || \pfink(\Delta(\rhi)))
\end{align}
From the {data-processing inequality}, it follows that 
\begin{subequations}\label{eq:data_proc_in_and_fin}
\begin{align}
    D_{KL} (\pin_{l}(\rhi)  || \pin_{l}(\Delta(\rhi))) & \leq D(\rhi || \Delta(\rhi)), \\
    D_{KL} (\pfink(\rhi)  || \pfink(\Delta(\rhi))) & \leq D(\rhi || \Delta(\rhi))
\end{align}
\end{subequations}
since the relative entropy is monotonically non-increasing under the action of quantum channels. For the former case, the quantum channel is the initial (virtual) energy measurement, whereas for the latter case, it is the composition of the CPTP map  $\Phi$ followed by final projective energy measurement. From Eq.~\eqref{eq:coh_DKL_in_plus_fin} and Eq.~\eqref{eq:data_proc_in_and_fin}, we get,
\begin{align}\label{eq:coh_data_process}
    D_{KL} (\pcoh^{l,k}(\rhi)  || \pcoh^{l,k}(\Delta(\rhi))) \leq 2 D(\rhi || \Delta(\rhi))
\end{align}
Identifying  $D(\rhi || \Delta(\rhi))$ with $C_{re}(\rhi)$ (see Eq.~\eqref{eq:Cre_as_relative_ent}) in Eq.~\eqref{eq:coh_data_process}, we obtain the second inequality of Eq.~\eqref{eq:CFD-bounds}.

The first inequality of~\eqref{eq:CFD-bounds} is rather trivial. It simply follows from the definition of $\mathfrak{D}_c (\rhi)$ (see Eq.~\eqref{eq:CFT-defn}) as the minimum value of $D_{KL} (\pcoh^{l,k}(\rhi)  || \pcoh^{l,k}(\rhi^{\mathcal{I}}))$ between a given coherent state $\rhi$ and {any} diagonal state $\rhi^{\mathcal{I}}$ subject to minimization over the diagonal states $\rhi^{\mathcal{I}}$. So naturally it cannot be greater than $D_{KL} (\pcoh^{l,k}(\rhi)  || \pcoh^{l,k}(\Delta(\rhi)))$.

\end{proof}

\subsection{Upper bounds of EFD}\label{appsec:upper-bounds:::subsec:EFD-proof}

We now prove the two sets of upper bounds \eqref{eq:EFD-upper-bound-1}, and \eqref{eq:EFD-upper-bound-2},
\begin{equation*}
  \mathfrak{D}_E (\rhi) \leq D_{KL}(\pent^{\lb, \kb} (\rhi) || \pent^{\lb, \kb} (\rhS)) \leq 2 D(\rhi || \rhS)
\end{equation*}
\begin{equation*}
   \mathfrak{D}_E (\rhi) \leq D_{KL}(\pent^{\lb, \kb} (\rhi) || \pent^{\lb, \kb} (\rhstar)) \leq 2 D_E (\rhi)
\end{equation*}

\begin{proof}
    The data-processing inequality gives us
    \begin{align}
       D_{KL} (\pinlb (\rhi) || \pinlb (\rhS) ) &\leq D(\rhi || \rhS) \\
       D_{KL} (\pfinkb (\rhi) || \pfinkb(\rhS) ) &\leq D(\rhi || \rhS) \\
    \end{align}
which gives
\begin{align}
    &D_{KL}(\pent^{\lb, \kb} (\rhi) || \pent^{\lb, \kb} (\rhS)) = \nonumber \\
    &D_{KL} (\pinlb (\rhi) || \pinlb(\rhS) ) +  D_{KL} (\pfinkb (\rhi) || \pfinkb(\rhS) )  \nonumber \\
    &\leq 2 D(\rhi || \rhS)
\end{align}
Thus proving second inequality of Eq.~\eqref{eq:EFD-upper-bound-1}.

The first inequality of Eq.~\eqref{eq:EFD-upper-bound-1} follows trivially from the definition Eq.~\eqref{eq:EFD-defn} of $\mathfrak{D}_E (\rhi)$.

We now turn to the proof of inequality relation Eq.~\eqref{eq:EFD-upper-bound-2}. The data-processing inequality gives us
 \begin{align}
       D_{KL} (\pinlb (\rhi) || \pinlb (\rhstar) ) &\leq D(\rhi || \rhstar) \\
       D_{KL} (\pfinkb (\rhi) || \pfinkb(\rhstar) ) &\leq D(\rhi || \rhstar) \\
    \end{align}
From which it immediately follows
\begin{align}
    &D_{KL}(\pent^{\lb, \kb} (\rhi) || \pent^{\lb, \kb} (\rhstar)) = \nonumber \\
    &D_{KL} (\pinlb (\rhi) || \pinlb(\rhstar) ) +  D_{KL} (\pfinkb (\rhi) || \pfinkb(\rhstar) )  \nonumber \\
    &\leq 2 D(\rhi || \rhstar)  = 2 D_E(\rhi)
\end{align}
This proves the second inequality of Eq.~\eqref{eq:EFD-upper-bound-2}. The first inequality of Eq.~\eqref{eq:EFD-upper-bound-2} follows trivially from the definition Eq.~\eqref{eq:EFD-defn} of $\mathfrak{D}_E (\rhi)$.

\end{proof}

\bibliography{references} 

\end{document}